\documentclass[journal, onecolumn,11pt, draftclsnofoot]{IEEEtran}

\usepackage{graphicx}
\usepackage{citesort}
\usepackage{amssymb}
\usepackage{amsfonts}
\usepackage{amsmath}
\usepackage{epsfig}
\usepackage{color}
\usepackage{fancybox}
\usepackage{textcomp}
\usepackage{multirow}
\usepackage{setspace}
\usepackage{psfrag}
\usepackage{booktabs}

\newtheorem{Theorem}{Theorem}

\newtheorem{Corollary}{Corollary}
\newtheorem{Proposition}{Proposition}

\newtheorem{Lemma}{Lemma}

\newtheorem{Definition}{Definition}

\newtheorem{Observation}{Observation}

    \def\ang#1{\mbox{$\langle #1 \rangle$}}
    \def\aang#1{\mbox{$\langle\!\langle #1 \rangle\!\rangle$}}
    \def\aangBig#1{\mbox{$\Big\langle\!\!\Big\langle #1 \Big\rangle\!\!\Big\rangle$}}
    \def\aangBigg#1{\mbox{$\Bigg\langle\!\!\!\Bigg\langle #1 \Bigg\rangle\!\!\!\Bigg\rangle$}}

    \def\Complex{{\rm\rule[.23ex]{.03em}{1.1ex}\kern-.3em{C}}}

    \newcommand{\be}{\begin{equation}} \newcommand{\ee}{\end{equation}}
    \newcommand{\bea}{\begin{eqnarray}} \newcommand{\eea}{\end{eqnarray}}
    \newcommand{\benum}{\begin{enumerate}} \newcommand{\eenum}{\end{enumerate}}

    \newcommand{\qb}{{\bf b}}

    \newcommand{\qu}{{\bf u}}
    \newcommand{\qv}{{\bf v}}
    \newcommand{\qw}{{\bf w}}
    \newcommand{\qx}{{\bf x}}
    \newcommand{\qy}{{\bf y}}
    \newcommand{\qz}{{\bf z}}

    \newcommand{\qA}{{\bf A}}
    \newcommand{\qB}{{\bf B}}

    \newcommand{\qI}{{\bf I}}

    \newcommand{\qQ}{{\bf Q}}
    \newcommand{\qR}{{\bf R}}
    
    \newcommand{\qT}{{\bf T}}
    \newcommand{\qU}{{\bf U}}
    \newcommand{\qV}{{\bf V}}
    \newcommand{\qW}{{\bf W}}
    \newcommand{\qX}{{\bf X}}

    \newcommand{\qzero}{{\bf 0}}
    \newcommand{\qone}{{\bf 1}}

    \newcommand{\qSigma}{{\boldsymbol \Sigma}}
    \newcommand{\qGamma}{{\boldsymbol \Gamma}}

    \newcommand{\tQ}{{\tilde{Q}}}

    \newcommand{\tm}{{\tilde{m}}}
    \newcommand{\tq}{{\tilde{q}}}

    \newcommand{\tqv}{{\tilde{\qv}}}
    \newcommand{\tqw}{{\tilde{\qw}}}
    \newcommand{\tqx}{{\tilde{\qx}}}
    \newcommand{\tqy}{{\tilde{\qy}}}

    \newcommand{\tqA}{{\tilde{\qA}}}

    \newcommand{\tqV}{{\tilde{\qV}}}
    
    \newcommand{\tqQ}{{\tilde{\qQ}}}

    \newcommand{\hQ}{{\hat{Q}}}
     \newcommand{\hm}{{\hat{m}}}
    \newcommand{\hchi}{{\hat{\chi}}}

    \newcommand{\oM}{{\overline{M}}}
    \newcommand{\oN}{{\overline{N}}}

    \newcommand{\bbC}{{\mathbb C}}

    \newcommand{\calN}{{\mathcal N}}
    
    \newcommand{\calO}{{\mathcal O}}
    \newcommand{\calR}{{\mathcal R}}

    \newcommand{\tr}{{\sf tr}}

    \newcommand{\Extr}{\operatornamewithlimits{\sf Extr}}
    
    \newcommand{\mse}{{\sf mse}}
    \newcommand{\vect}{{\sf vec}}

    \newcommand*{\argmin}{\operatornamewithlimits{argmin}\limits}
    
   \newcommand{\dRe}{{\rm Re}}
   \newcommand{\dIm}{{\rm Im}}

    \newcommand{\sfQ}{{\sf Q}}

    \newcommand{\rmd}{{\rm d}}
    \newcommand{\rmD}{{\rm D}}
    \newcommand{\sfj}{{\sf j}}

    \newcommand{\sfc}{{\sf c}}
    \newcommand{\sfr}{{\sf r}}

    \newcommand{\sfZ}{{\sf Z}}

    \newcommand{\aag}{\acute{g}}

\begin{document}

\title{On Sparse Vector Recovery Performance in Structurally Orthogonal Matrices via LASSO}

\author{Chao-Kai Wen$\stackrel{\ddag}{,}$ %~\IEEEmembership{Member,~IEEE,}
Jun Zhang$\stackrel{\sharp}{,}$
Kai-Kit Wong$\stackrel{\S}{,}$ %~\IEEEmembership{Senior Member,~IEEE,}
Jung-Chieh Chen$\stackrel{\Diamond}{,}$~and~Chau Yuen$^\sharp$
\thanks{$^\ddag$C. Wen is with the Institute of Communications Engineering, National Sun Yat-sen University, Kaohsiung 804, Taiwan. E-mail: ${\rm chaokai.wen@mail.nsysu.edu.tw}$.}% <-this % stops a space
\thanks{$^\sharp$J. Zhang and C. Yuen are with Engineering Product Development, Singapore University of Technology and Design, Singapore. E-mail: ${\rm\{zhang\_jun,yuenchau\}@sutd.edu.sg}$.}
\thanks{$^\S$K. Wong is with the Department of Electronic and Electrical Engineering, University College London, London, United Kingdom. E-mail: ${\rm kai\mbox{-}kit.wong@ucl.ac.uk}$.}
\thanks{$^\Diamond$J.-C. Chen (Corresponding author) is with the Department of Optoelectronics and Communication Engineering, National Kaohsiung Normal University, Kaohsiung 802, Taiwan. E-mail: ${\rm jcchen@nknucc.nknu.edu.tw}$.} }

%\markboth{IEEE Transactions on Wireless Communications,~Vol.~XX, No.~XX, XXX~2013}%
%{Shell \MakeLowercase{\textit{et al.}}: Bare Demo of IEEEtran.cls for Journals}

\maketitle

\begin{abstract}
In this paper, we consider a compressed sensing problem of reconstructing a sparse signal from an undersampled set of noisy linear measurements. The regularized least squares or least absolute shrinkage and selection operator (LASSO) formulation is used for signal estimation. The measurement matrix is assumed to be constructed by concatenating several randomly orthogonal bases, referred to as structurally orthogonal matrices. Such measurement matrix is highly relevant to large-scale compressive sensing applications because it facilitates fast computation and also supports parallel processing. Using the replica method from statistical physics, we derive the mean-squared-error (MSE) formula of reconstruction over the structurally orthogonal matrix in the large-system regime. Extensive numerical experiments are provided to verify the analytical result. We then use the analytical result to study the MSE behaviors of LASSO over the structurally orthogonal matrix, with a particular focus on performance comparisons to matrices with independent and identically distributed (i.i.d.) Gaussian entries. We demonstrate that the structurally orthogonal matrices are at least as well performed as their i.i.d.~Gaussian counterparts, and therefore the use of structurally orthogonal matrices is highly motivated in practical applications.
\end{abstract}


\section*{\sc I. Introduction}
Sparse signal reconstruction problems appear in many engineering fields. In most applications, signals are often measured from an undersampled set of noisy linear transformations. Typically, the problem of interest is the reconstruction of a \emph{sparse} signal $\qx^{0} \in \bbC^{\oN}$ from a set of $\oM (\leq \oN)$ noisy measurements $\qy \in \bbC^{\oM}$ which is given by
\begin{equation}\label{eq:sysModel}
\qy= \qA\qx^{0} + \sigma_0\qw,
\end{equation}
where $\qA \in \bbC^{\oM \times \oN}$ is the measurement matrix, and $\sigma_0 \qw \in\bbC^{\oM}$ is the noise vector with $\sigma_0$ representing the noise magnitude. This problem has arisen in many areas, such as signal processing, communications theory, information science, and statistics, and is widely known as {\em compressive sensing} \cite{Candes-05IT,Donoho-06IT}.

In the past few years, many recovery algorithms have been proposed, see \cite{Tropp-10ProcIEEE,Hayashi-13IEICE} for a recent exhaustive list of the algorithms. One popular suboptimal and low-complexity estimator is $\ell_1$-regularized least-squares (LS), a.k.a.~least absolute shrinkage and selection operator (LASSO) \cite{Tibshirani-96JRSS}, which seeks $\qx^{0}$ by
\begin{equation} \label{eq:RLS}
   \hat{\qx} =  \argmin_{\qx \in \bbC^N} \left\{ \frac{1}{\lambda}\| \qy - \qA\qx\|_2^2 + \|\qx\|_1 \right\}.
\end{equation}
In (\ref{eq:RLS}), $\lambda > 0 $ is a design parameter, and the \emph{complex}\footnote{In the real-valued setting, the $\ell_1$-norm is defined as $\|\qx\|_1\triangleq\sum_{n} |x_n|$, which is different from the complex $\ell_1$-norm. A simple extension of LASSO to the complex setting is to consider the complex signal and measurements as a $2\oN$-dimensional real-valued signal and $2\oM$-dimensional real-valued measurements, respectively. However, several papers (e.g., \cite{Maleki-13IT} and the references therein) have shown that LASSO based on the complex $\ell_1$-norm is superior to the simple real-valued extension when the real and imaginary components of the signals tend to either zero or nonzero simultaneously. Therefore, we consider LASSO using the the complex $\ell_1$-norm definition of (\ref{eq:RLS}) rather than the simple real-valued extension of LASSO.} $\ell_1$-norm is defined as
\begin{equation}
\|\qx\|_1\triangleq \sum_{i=1}^{\oN} |x_i| = \sum_{i=1}^{\oN} \sqrt{ (\dRe\{x_i\})^2 + (\dIm\{x_i\})^2}.
\end{equation}
The optimization problem of (\ref{eq:RLS}) is convex, and there are various fast and efficient solvers proposed. For example, the proximal gradient method in \cite[Section 7.1]{Parikh-14FTinOpt} resolves (\ref{eq:RLS}) by iteratively performing
\begin{equation} \label{eq:proxGrad}
    \hat{\qx}^{t+1} := \eta\Big( \hat{\qx}^{t} - \varsigma^{t} \qA^H(\qA\hat{\qx}^{t}-\qy), \varsigma^{t} \Big),
\end{equation}
where $t$ is the iteration counter, $\varsigma^{t} > 0 $ is the chosen step size, and
\begin{equation}
    \eta(x, \varsigma) \triangleq \frac{x}{|x|} \left( |x| - \varsigma \right)_{+}
\end{equation}
is a soft-thresholding function in which $(a)_{+} = a$ if $a > 0$ and is $0$ otherwise.

Evaluating $\qA^H(\qA\hat{\qx}^{t}-\qy)$ requires one matrix-vector multiplication by $\qA$ and another by $\qA^H$, plus a (negligible) vector addition. The complexity for evaluating the soft-thresholding function $\eta$ is negligible. This kind of iterative thresholding algorithm requires few computations per-iteration, and therefore enables the application of LASSO in large-scale problems.

Much of the theoretical work on (\ref{eq:RLS}) has focused on studying how aggressively a sparse signal can be undersampled while still guaranteeing perfect signal recovery. The existing results include those based on the restricted isometry property (RIP) \cite{Candes-05IT,Candes-08}, polyhedral geometry \cite{Donoho-05PNAS,Donoho-10ProcIEEE}, message passing \cite{Donoho-09PNAS}, and the replica method \cite{Kabashima-09JSM,Ganguli-10PRL,Rangan-12TIT}. Although RIP provides sufficient conditions for sparse signal reconstruction, the results provided by RIP analysis are often conservative in practice. In contrast, using combinational geometry, message passing, or the replica method, it is possible to compute the exact necessary and sufficient condition for measuring the sparsity-undersampling tradeoff performance of (\ref{eq:RLS}) in the limit $\oN \to \infty$. However, the theoretical work largely focused on the case of having a measurement matrix $\qA$ with independent and identically distributed (i.i.d.) entries. A natural  question would be ``{\em how does the choice of the measurement matrix affect the typical sparsity-undersampling tradeoff performance?}''.

There are strong reasons to consider different types of measurement matrix. Although the proximal gradient method performs efficiently in systems of medium size,
the implementation of (\ref{eq:proxGrad}) will become prohibitively complex if the signal size is very large. This is not only because performing
(\ref{eq:proxGrad}) requires matrix-vector multiplications up to the order of $O(\oM\oN)$ but it also requires a lot of memory to store the measurement matrix.
There is strong desire to consider special forms of measurement matrix permitting faster multiplication process and requiring less memory. One such example is the
randomly generated discrete Fourier transform (DFT) matrices or discrete cosine transform (DCT) matrices \cite{Do-08ICASSP,Barbier-13ArXiv,Javanmard-12ISIT}. Using
DFT as the measurement matrix, fast Fourier transform (FFT) can be used to perform the matrix multiplications at complexity of $O(\oN\log_2\oN)$ and the
measurement matrix is \emph{not} required to be stored. The entries of a DFT matrix are however not i.i.d..

In the noiseless setup (i.e., $\sigma =0$), it has been revealed that the measurement matrix enjoys the so-called universality property; that is, measurement matrices with i.i.d.~ensembles and rotationally invariant (or row-orthonormal) ensembles exhibit the same recovery capability (or the phase transition) \cite{Donoho-05PNAS,Donoho-09MPS,Bayati-12ArXiv,Kabashima-09JSM}. The universality phenomenon is further extended to the measurement matrices which are constructed by concatenating several randomly square orthonormal matrices \cite{Kabashima-12JSM}.

Although the universality phenomenon of LASSO is known for a broad class of measurement matrices in the noiseless setup, little progress has been made in the
practical \emph{noisy} setting. In the noisy setting, perfect recovery is rare so we are interested in the (average) mean squared error (MSE) of reconstruction
defined by $\oN^{-1} \aang{ \| \qx^{0} - \hat{\qx} \|_2^2 }_{\qw,\qx^{0}}$, where $\aang{\cdot}_{\qw,\qx^{0}}$ denotes the average with respect to $\qw$ and
$\qx^{0}$. In \cite{Tulino-13IT}, an analytical expression for MSE in LASSO reconstruction was obtained when the measurement matrix $\qA$ is a row-orthonormal
matrix generated uniformly at random. Nevertheless, the emphasis of \cite{Tulino-13IT} was in support recovery rather than the MSE of reconstruction. It was not
until very recently that the superiority of row-orthonormal measurement matrices over their i.i.d.~Gaussian counterparts for the noisy sparse recovery problem was
revealed in \cite{Vehkapera-arXiv13}.\footnote{In fact, the significance of orthogonal matrices under other problems (e.g., code-division multiple-access (CDMA)
and multiple-input multiple-output (MIMO) systems) has been pointed out much earlier in \cite{Takeda-06EPL,Hatabu-09PRE,Kitagawa-10CN}.} This characteristics is in
contrast to the noiseless setup mentioned above. Meanwhile, the authors of \cite{Oymak-14ISIT} supported the similar argument and further argued that one can still
claim universality in the noisy setup if we restrict the measurement matrices to similar row-orthonormal type. These arguments showed that the choice of
measurement matrices does have an impact in the MSE of reconstruction when noise is present. Despite these previous works, the study of LASSO in the case of
orthonormal measurement matrices remains incomplete, for the following reasons.

\begin{figure}
\begin{center}
\resizebox{4in}{!}{%
\includegraphics*{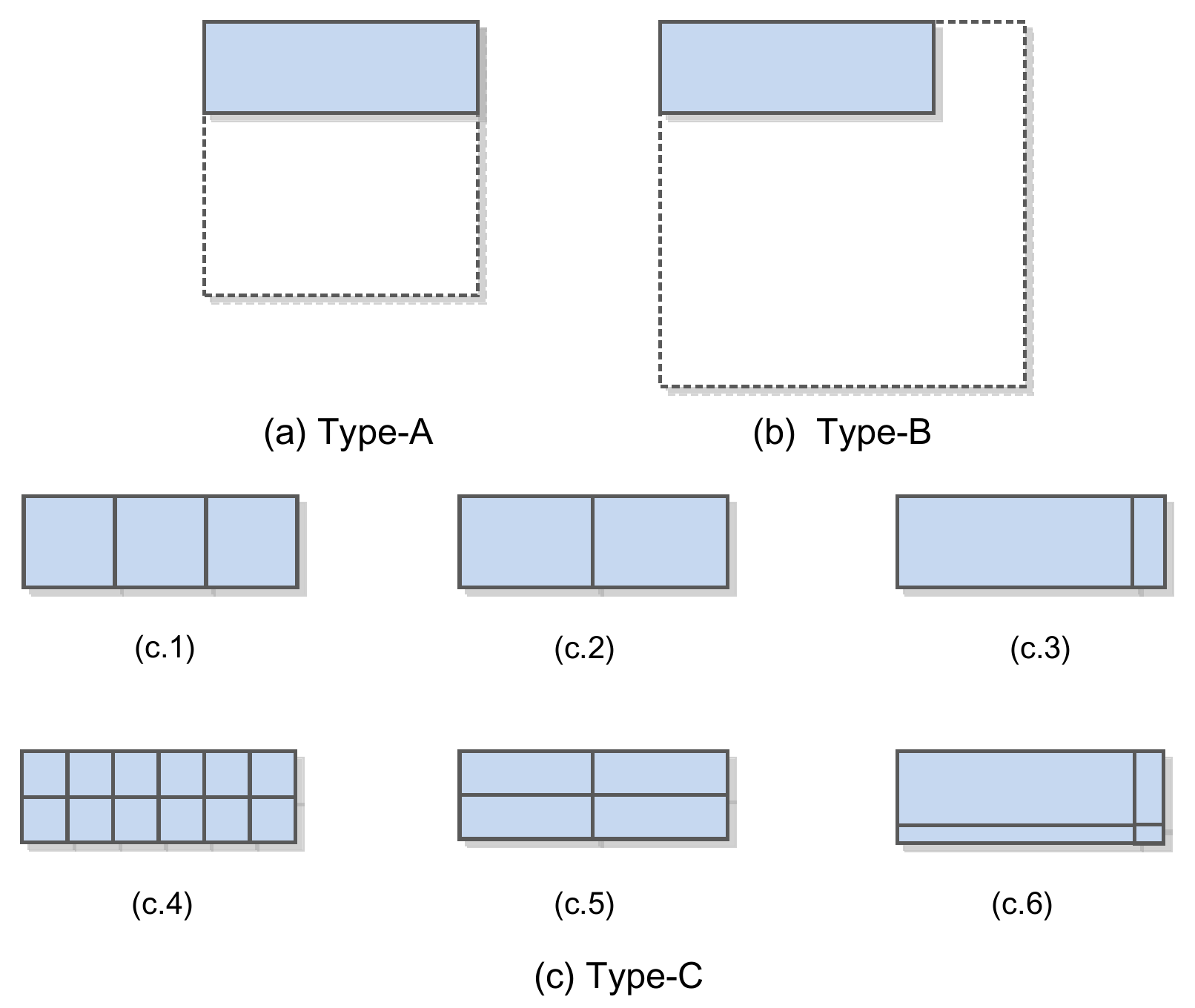} }%
\caption{Examples of structurally random orthogonal matrices.}\label{fig:M-Orthogonal}
\end{center}
\end{figure}

First, in many applications of interest, the measurement matrices are constructed by selecting a set of columns and rows from a standard orthonormal matrix as depicted in Figure \ref{fig:M-Orthogonal}(b), which we call it ``Type-B'' matrix. Let a standard orthonormal matrix be an $N \times N$ unitary matrix. Then we have $\oM < \oN < N$ in a Type-B matrix. Note that it is also possible to obtain row-orthonormal matrices by selecting a set of rows from orthonormal matrices rather than selecting a set of rows \emph{and} columns, and in this case, we refer to such row-orthonormal matrix as ``Type-A'' matrix, see Figure \ref{fig:M-Orthogonal}(a). One prominent application of Type-B matrices is the sparse channel estimator in orthogonal frequency-division multiplexing (OFDM) systems \cite{Bajwa-10Proc}. In that case, $\qx^{0}$ represents a time-domain channel vector and $\qA$ is a partial DFT matrix. Another popular application arises in compressive sensing imaging, where randomly sampling a DFT matrix from its row is common practice \cite{Do-08ICASSP}. Nonetheless, \emph{additional} selection in columns is often needed because the signal size would be smaller than the size of available FFT operators\footnote{Though FFT is generally applicable for any size, many available FFT operators in DSP chips are of power-of-two sizes \cite{Mitra-DSP}.} and the signal will be modified by zero-padding to fit the available FFT size. In that case, the measurement matrix corresponds to the matrix formed by selecting a set of columns from row-orthonormal matrices. While Type-B measurement matrices are widely employed in a vast number of sparse recovery problems, surprisingly, little is known on the LASSO performance based on such measurement matrices.

Measurement matrix that is constructed by concatenating several randomly chosen orthonormal bases is another common type, which is referred to as ``Type-C'' matrix. Such construction can have several variations as shown in Figure \ref{fig:M-Orthogonal}(c) due to certain implementation considerations \cite{Do-08ICASSP,Fowler-12FtSP}. For example, we can exploit parallelism or distributed computation of (\ref{eq:proxGrad}) by using a parallel matrix-vector multiplication. In that case, each sub-block of the measurement matrix would be taken from a partial block of scrambled DFT matrix. In this context, the authors of \cite{Vehkapera-arXiv13} demonstrated that Type-A matrix and Type-C.1 matrix (constructed by concatenating several randomly \emph{square} orthonormal matrices) have the same performance. Except for \cite{Vehkapera-arXiv13}, however, little progress has been made on this type of measurement matrix.

In this paper, we aim to provide analytical characterization for the performance of LASSO under such measurement matrices. In particular, we  derive the MSE of LASSO in the general Type-C setup by using the replica method from statistical physics as in \cite{Kabashima-09JSM,Tanaka-10ISIT,Kabashima-12JSM,Vehkapera-arXiv13,Tulino-13IT,Rangan-12TIT}. Our MSE result encompasses Type-A and Type-B matrices as special cases. Then we compare their performances and behaviors with those for random i.i.d.~Gaussian matrices. We will show that all the structurally orthogonal matrices (including Types A--C) perform at least as well as random i.i.d.~Gaussian matrices over arbitrary setups.\footnote{To make a fair comparison among different setups, we have properly normalized their energy consumption.}

Specifically, we have made the following technical contributions:
\begin{itemize}
\item We show that Type-A matrix has the best MSE performance out of all other types of structurally orthogonal matrices and performs significantly better than the i.i.d.~Gaussian matrices.
\item In contrast to Type-A matrices, the row-orthogonality in Type-B is no longer preserved if $\oN < N$.
The MSE performance of Type-B matrices degrades with decreasing the ratio of $\oN/N$ while they still perform at least as good as their random i.i.d.~Gaussian
counterparts.
\item We show that Type-A, Type-C.1, and Type-C.2 matrices have the same MSE performance. Specifically, \emph{horizontally} concatenating multiple row-orthonormal matrices have the same MSE performance as its single row-orthonormal counterpart. This argument extends the result of
\cite{Vehkapera-arXiv13} to the case of concatenating multiple row-orthonormal matrices. Further, we reveal that the measurement matrices formed by concatenating several randomly orthonormal bases in \emph{vertical} direction result in significant degradation. For example, Type-C.4 and Type-C.5 matrices have the worst performance among Type-C matrices although they are at least as good as their random i.i.d.~Gaussian counterparts.
\end{itemize}

The remainder of the paper is organized as follows. In Section II, we present the problem formulation including fundamental definitions of the structurally orthogonal matrices. In Section III, we provide the theoretical MSE results of LASSO based on the structurally orthogonal matrices. Simulations and discussions are presented in Section IV and the main results are summarized in Section V.

{\em Notations}---Throughout this paper, for any matrix $\qA$, $[\qA]_{i,j}$ refers to the $(i,j)$th entry of $\qA$, $\qA^T$ denotes the transpose of $\qA$, $\qA^H$ denotes the conjugate transpose of $\qA$, $\sqrt\qA$ (or $\qA^\frac{1}{2}$) denotes the principal square root of $\qA$, $\tr(\qA)$ denotes the trace of $\qA$, ${\tt vec}(\qA)$ is the column vector with entries being the ordered stack of columns of $\qA$. Additionally, $\qI_n$ denotes an $n$-dimensional identity matrix, ${\bf 0}$ denotes a zero matrix of appropriate size, ${\bf 1}_n$ denotes an $n$-dimensional all-one vector, $\|\cdot\|_2$ denotes the Euclidean norm, ${\sf I}_{\{\mbox{statement}\}}$ denotes the indicator of the statement, $\aang{\cdot}_X$ represents the expectation operator with respect to $X$, ${\rm log}(\cdot)$ is the natural logarithm, $\delta(\cdot)$ denotes Dirac's delta, $\delta_{i,j}$ denotes Kronecker's delta, $\Extr_{x}\{f(x)\}$ represents the extremization of a function $f(x)$ with respect to $x$, $\sfQ(\xi) \triangleq \frac{1}{\sqrt{2 \pi}} \int_{\xi}^{\infty} e^{x^2/2} \rmd x$ is the standard $\sfQ$-function,
and $\otimes$ denotes the Kronecker product.
We say that the complex random variable $Z$ is a standard Gaussian, if its density function is given by $\calN(z) \triangleq \frac{1}{\pi} e^{-|z|^2}$. That is, the standard complex Gaussian is the circularly-symmetric complex Gaussian with zero mean and unit variance.
Finally, $D \qz$ denotes the complex-valued Gaussian integration measure; i.e., for an $n \times 1$ vector, $\qz$, we have
\begin{equation*}
D\qz = \prod_{i=1}^n \frac{d \dRe\{z_i\} d \dIm\{z_i\} }{\pi} e^{-(\dRe\{z_i\})^2-(\dIm\{z_i\})^2}
\end{equation*}
where $\dRe\{ \cdot \}$ and $\dIm\{ \cdot \}$ extract the real and imaginary components, respectively.

\begin{figure}
\begin{center}
\resizebox{4.5in}{!}{%
\includegraphics*{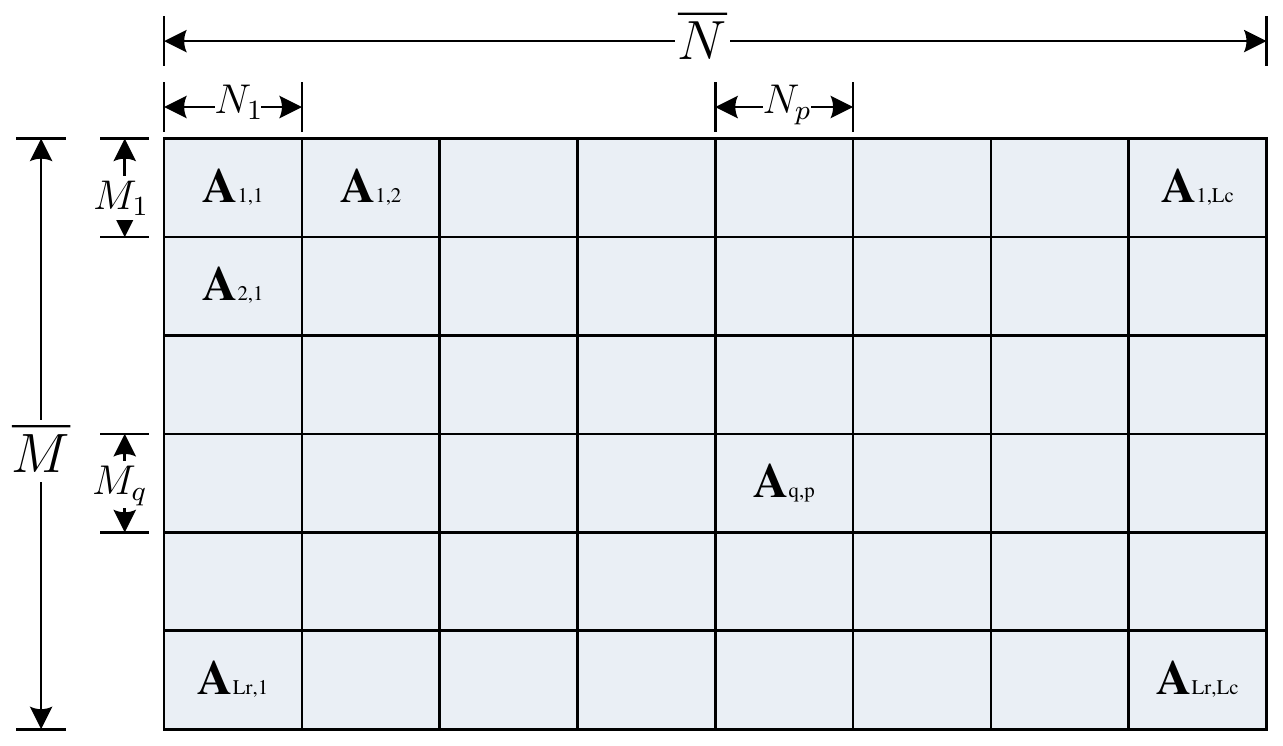} }%
\caption{An example of a structurally random matrix, where $L_c= 8$ and $L_r=6$. Each block is obtained from an independent standard $N \times N$ orthonormal matrix by selecting $M_q$ rows and $N_p$ columns at random.}\label{fig:Type-C-Matrix}
\end{center}
\end{figure}

%\vspace{-.1in}
\section*{\sc II. Problem Formulation}
We consider the sparse signal recovery setup in (\ref{eq:sysModel}), where $\qw$ is assumed to be the standard complex Gaussian noise vector. In addition, let us suppose that
\begin{equation}\label{eq:px0}
    P_0(\qx^{0}) = \prod_{n=1}^{\oN} P_0(x_n),
\end{equation}
where $P_0(x_n) = (1-\rho_x)\delta(x_n) + \rho_x\calN(x_n)$ for $n = 1,\ldots,\oN$, and $\rho_x \in [0,\,1]$ is the fraction of non-zero entries in $\qx^{0}$. That is, the elements of $\qx^{0}$ are sparse and are i.i.d.~generated according to $P_0(x_n)$.

For generality, we consider the measurement matrix $\qA$ made of different blocks as outlined in Fig.~\ref{fig:Type-C-Matrix}, which we refer to it as Type-C matrix. The structurally random matrix was also considered by \cite{Krzakala-12JSM,Barbier-13ArXiv} in the context of compressive sensing for different purposes. In the setup, $\qA \in \bbC^{\oM \times \oN}$ is constructed by vertical and horizontal concatenation of $L_r\times L_c$ blocks as
\begin{equation} \label{eq:Amatrix}
 \qA = \left[
 \begin{array}{ccc}
 \qA_{1,1} & \cdots & \qA_{1,L_c} \\
 \vdots & \ddots  & \vdots \\
 \qA_{L_r,1} & \cdots & \qA_{L_r,L_c}
 \end{array}
 \right],
\end{equation}
where each $\qA_{q,p} \in \bbC^{M_q \times N_p}$ is drawn independently from the Haar measure of $N \times N$ random matrix (referred to as the {\em standard} orthonormal matrix in this paper). To shape $\qA_{q,p}$ in an $M_q \times N_p$-dimensional matrix, we randomly select $M_q$ rows and $N_p$ columns from the standard orthonormal matrix. We denote $\mu_p = N_p/N$ and $\nu_q = M_q/N$ the ``column selection rate'' and ``row selection rate'', respectively. Also, we define $\mu \triangleq \sum_{p} \mu_p = \oN/N$ and $\nu \triangleq\sum_{q} \nu_q = \oM/N$. To make the setup more flexible, we assume that for the $(q,p)$th subblock, the standard orthonormal matrix has been multiplied by $\sqrt{R_{q,p}}$. By setting the values of $R_{q,p}$ appropriately, each block can be made either only zeros or a partial orthonormal matrix.

Corresponding to the measurement matrix made of different blocks, the $\oN$ variables of $\qx^{0}$ are divided into $L_c$ blocks $\{ \qx_p : p=1,\ldots,L_c\}$ with $N_p$ variables in each block. Meanwhile, the $\oM$ measurements $\qy$ are divided into $L_r$ blocks $\{ \qy_q : q=1,\ldots,L_r\}$ with $M_q$ measurements in each block. Note that we have $\oM = \sum_{q=1}^{L_r} M_q$ and $\oN = \sum_{p=1}^{L_c} N_p$. The measurement ratio of the system is given by $\alpha = \oM/\oN$.

\section*{\sc III. Analytical Results}
To facilitate our analysis based on the tools in statistical mechanics, we use the approach introduced in \cite{Kabashima-12JSM,Vehkapera-arXiv13} to reformulate the $\ell_1$-regularized LS problem (\ref{eq:RLS}) in a probabilistic framework. Suppose that the posterior distribution of $\qx$ follows the distribution
\begin{equation} \label{eq:postDis_RLS}
P_{\beta}(\qx|\qy) = \frac{1}{\sfZ_{\beta}(\qy,\qA)} e^{-\beta\left(\frac{1}{\lambda}\|\qy-\qA\qx\|_2^2 + \|\qx\|_1 \right)},
\end{equation}
where $\beta$ is a constant and
\begin{equation} \label{eq:partFun_B}
\sfZ_{\beta}(\qy,\qA) = {\int \rmd\qx}\, e^{-\beta\left(\frac{1}{\lambda}\|\qy-\qA\qx\|_2^2 + \|\qx\|_1 \right)}
\end{equation}
is the partition function (or normalization factor) of the above distribution function. Given the posterior probability of (\ref{eq:postDis_RLS}), the Bayes way of
estimating $\qx$ is given by \cite{Poor-94BOOK}
\begin{equation} \label{eq:estx}
    \ang{\qx}_{P_{\beta}} = {\int\rmd\qx} \, \qx P_{\beta}(\qx|\qy) .
\end{equation}
As $\beta \rightarrow \infty$, the posterior mean estimator (\ref{eq:estx}) condenses to the global minimum of (\ref{eq:RLS}), i.e., $\ang{\qx}_{P_{\beta}} = \hat{\qx}$.

In (\ref{eq:estx}), $\ang{\qx}_{P_{\beta}}$ (or equivalently $\hat{\qx}$) is estimated from $\qy$ given that $\qA$ is perfectly known. Clearly, $\hat{\qx}$ depends on $\qy$
and thus is \emph{random}. We are thus interested in the (average) MSE of $\hat{\qx}$ given by
\begin{align}
\mse &= \oN^{-1} \aangBig{ \| \qx^{0} - \hat{\qx} \|_2^2 }_{\qy} \nonumber \\
         &= \rho_x  - 2 \oN^{-1} \aangBig{ \dRe\left\{ \hat{\qx}^H \qx^{0} \right\} }_{\qy} + \oN^{-1} \aangBig{ \hat{\qx}^H \hat{\qx} }_{\qy},\label{eq:defMSE}
\end{align}
where $\aang{ \cdot }_{\qy}$ denotes an average over $\qy$. Specifically, we define
\begin{equation}
\aang{ f(\qy) }_{\qy} \triangleq \int \rmd \qy \int \rmd \qx^{0} f(\qy) P(\qy|\qx^{0}))P_0(\qx^{0}),
\end{equation}
where $P_0(\qx^{0})$ is defined by (\ref{eq:px0}), and
\begin{equation} \label{eq:postDis_sys}
    P(\qy|\qx^{0}) = \frac{1}{(\pi\sigma_0^2)^{\oM}} e^{-\frac{1}{\sigma_0^2}\|\qy-\qA\qx^{0}\|_2^2}
\end{equation}
is the conditional distribution of $\qy$ given $\qx^{0}$ under (\ref{eq:sysModel}). Our aim of this paper is to derive an analytical result for $\mse$.

In the analysis of $\mse$, we consider $N \rightarrow \infty$, while keeping $\mu_p = N_p/N$ and $\nu_q = M_q/N$ fixed and finite for $p=1,\ldots,L_c$ and $q=1,\ldots,L_r$. For convenience, we refer to this large dimensional regime simply as $N \rightarrow \infty$. Notice that the MSE depends on the measurement matrix $\qA$. However, in the large regime $N \rightarrow \infty$, we expect (or assume) that the average MSE appears to be self-averaging. That is, the MSE for any typical realization of $\qA$ coincides with its average over $\qA$.

From  (\ref{eq:defMSE}), the posterior distribution $P_{\beta}$ plays a role in the MSE. In statistical mechanics, the key for finding the MSE is through
computing the partition function, which is the marginal of $P_{\beta}$, or its logarithm, known as {\em free entropy}. Following the argument of
\cite{Kabashima-09JSM,Vehkapera-arXiv13}, it can be shown that $\mse$ is a saddle point of the free entropy. Thanks to the self-averaging property in the large
dimensional regime, we therefore compute $\mse$ by computing the \emph{average} free entropy
\begin{equation}\label{eq:FreeEn}
    \Phi = \lim_{\beta, N \rightarrow \infty}\frac{1}{\beta N} \aang{\log \sfZ_{\beta}(\qy,\qA)}_{\qy,\qA}.
\end{equation}
The similar manipulation has been used in many different settings, e.g., \cite{Rangan-12TIT,Krzakala-12JSM,Tulino-13IT,Kabashima-12JSM,Vehkapera-arXiv13}. The analysis of (\ref{eq:FreeEn}) is unfortunately still difficult. The major difficulty in (\ref{eq:FreeEn}) lies in the expectations over $\qy$ and $\qA$. We can, nevertheless, greatly facilitate the mathematical derivation by rewriting $\Phi$ as \cite{Kabashima-09JSM,Vehkapera-arXiv13}
\begin{equation}\label{eq:avgFreeEn}
\Phi = \lim_{\beta, N \rightarrow \infty}\frac{1}{\beta N} \lim_{\tau \rightarrow 0}\frac{\partial}{\partial \tau}\log\aang{\sfZ_{\beta}^{\tau}(\qy,\qA)}_{\qy,\qA},
\end{equation}
in which we have moved the expectation operator inside the log-function. We first evaluate $\aang{\sfZ_{\beta}^{\tau}(\qy,\qA)}_{\qy,\qA}$ for an integer-valued $\tau$, and then generalize it for any positive real number $\tau$. This technique is known as the replica method, which emerged from
the field of statistical physics \cite{Edwards-75JPF,Nishimori-01BOOK} and has recently been successfully applied to information/communications theory literature \cite{Kabashima-09JSM,Tanaka-10ISIT,Kabashima-12JSM,Krzakala-12JSM,Vehkapera-arXiv13,Tulino-13IT,Rangan-12TIT,Tanaka-02IT,Moustakas-03TIT,Guo-05IT,Muller-03TSP,Wen-07TCOM,Hatabu-09PRE,Takeuchi-13TIT,Girnyk-14TWC}.

Details of the replica calculation are provided in Appendix A. We here intend to give an intuition on the final analytical results (i.e., Proposition \ref{Pro1} to be shown later). Note that the approach presented here is slightly different from that in Appendix A. Basically, the replica analysis allows us to understand the characteristics of the errors made by LASSO by looking at the signal reconstruction via an equivalent scalar version of the linear system (\ref{eq:sysModel}):
\begin{equation} \label{eq:EqScalSysModel}
 y_p = \hm_{p}x_p^{0} + \sqrt{\hchi_{p}}z_{p},
\end{equation}
where the subscript $p$ indicates that the equivalent linear system characterizes the signal in block $p$, i.e. $\qx_p$, and there are $N_p$ parallel equivalent linear systems of (\ref{eq:EqScalSysModel}) in block $p$; the parameters $(\hm_{p},\hchi_{p})$ are arisen in the replica analysis to be given later in Proposition \ref{Pro1}; $x_p^{0}$ is a random signal generated according to the distribution $P_0(x)$, $z_{p}$ is standard complex Gaussian; and $y_p$ is the effective measurement.

In particular, our analysis shows that the characteristics of LASSO output corresponding to the signal, $\qx_p$, can be analyzed via the LASSO output of the signal $x_p^{0}$ through the effective measurement $y_p$, where $\hm_{p}$ and $\hchi_{p}$ play the role of the \emph{effective} measurement gain and \emph{effective} noise level. Therefore, following (\ref{eq:RLS}), the recovery of $x_p^{0}$ from $y_p$ by LASSO becomes
\begin{equation} \label{eq:RLS_EqScalSysModel}
   \hat{x}_p =  \argmin_{x_p \in \bbC} \left\{ \frac{1}{\hm_{p}}| y_p - \hm_{p}x_p|^2 + |x_p| \right\}.
\end{equation}
Using \cite[Lemma V.1]{Maleki-13IT}, the optimal solution $\hat{x}_{p}$ of (\ref{eq:RLS_EqScalSysModel}) reads
\begin{equation} \label{eq:optOfxp}
 \hat{x}_p
 = \frac{\left( |y_p| - \frac{1}{2}\right)_{+} \frac{y_p}{|y_p|}}{\hm_{p}}.
\end{equation}
Note that $\hat{x}_p$ depends on $y_p$ and is therefore random. Then the MSE of $\hat{x}_p$ is given by $\aang{ | x_p^{0} - \hat{x}_p |^2 }_{y_p} = \rho_x  - 2
\aang{\dRe\{ \hat{x}_p^* x_p^{0} \} }_{y_p} + \aang{ |\hat{x}_p|^2 }_{y_p}$, where $\aang{ \cdot }_{y_p}$ denotes an average over $y_p$ with
\begin{equation} \label{eq:postDis_EqScalSysModel}
    P(y_p|x_p^{0}) = \frac{1}{\pi \hchi_{p}} e^{-\frac{1}{\hchi_{p}}|y_p-\hm_{p}x_p^{0}|^2}.
\end{equation}
As there are $N_p$ parallel equivalent systems in block $p$, the MSE of LASSO reconstruction in group $p$ is given by
\begin{align}
\mse_p &= \frac{N_p}{\oN} \aangBig{ \left| x_p^{0} - \hat{x}_p \right|^2 }_{y_p} \nonumber \\
 &= \frac{\mu_p}{\mu} \aangBig{ \left| x_p^{0} - \hat{x}_p \right|^2 }_{y_p} \nonumber \\
 &= \frac{1}{\mu} \left(\mu_p\rho_x  - 2 \mu_p\aangBig{ \dRe\left\{ \hat{x}_p^* x_p^{0} \right\} }_{y_p} + \mu_p\aangBig{ \left|\hat{x}_p\right|^2 }_{y_p}\right) \nonumber \\
 &=  \frac{1}{\mu} \left(\mu_p\rho_x  - 2m_p + Q_{p}\right),
\end{align}
where the second equality is due to $\mu_p = N_p/N$ and $\mu = \oN/N$, and the last equality follows from the fact that $m_p \triangleq \mu_p \aang{ \dRe\{ \hat{x}_p^* x_p^{0} \}}_{y_p}  $ and $Q_{p} \triangleq \mu_p\aang{ |\hat{x}_p|^2 }_{y_p}$. Using (\ref{eq:optOfxp}) and (\ref{eq:postDis_EqScalSysModel}) and following the steps of \cite[(349)--(357)]{Tulino-13IT}, one can get the analytical expressions of $m_p$ and $Q_{p}$, and then result in an analytical expression of $\mse_p$. We summarize the results in the following proposition.

\begin{Proposition} \label{Pro1}
Consider a Type-C matrix being the measurement matrix. Let $\mse_p$ denote the MSE of LASSO reconstruction in block $p=1,\ldots,L_c$, and define
\begin{subequations} \label{eq:defgc}
\begin{align}
g_{\sfc}(\zeta) &\triangleq \zeta e^{-\frac{1}{4\zeta}} - \sqrt{\pi\zeta} \sfQ\left(\frac{1}{\sqrt{2\zeta}}\right), \\
\aag_{\sfc}(\zeta) &\triangleq e^{-\frac{1}{4\zeta}} - \sqrt{\frac{\pi}{4\zeta}} \sfQ\left(\frac{1}{\sqrt{2\zeta}}\right).
\end{align}
\end{subequations}
Then as $N \rightarrow \infty$, the average MSE over the entire vector becomes
\begin{equation}
 \mse = \sum_{p=1}^{L_c} \mse_{p},
\end{equation}
where $\mse_{p} = (\mu_{p}\rho_x - 2m_{p}+ Q_{p})/\mu $ with
\begin{subequations} \label{eq:mQ_Pro1}
\begin{align}
m_{p} &= \mu_p\rho_x \aag_{\sfc}(\hm_{p}^2+\hchi_{p}) , \\
Q_{p} &= \mu_p \left(\frac{1-\rho_x}{\hm_{p}^2} g_{\sfc}(\hchi_{p}) + \frac{\rho_x}{\hm_{p}^2} g_{\sfc}(\hm_{p}^2+\hchi_{p})\right).
\end{align}
\end{subequations}
In (\ref{eq:mQ_Pro1}), we have defined
\begin{equation} \label{eq:hm_Pro1}
\hm_{p} \triangleq  \frac{\sum_{q=1}^{L_r}\Delta_{q,p}}{\chi_{p}},
\end{equation}
where
\begin{align}
\Delta_{q,p} &= \nu_{q} \frac{\frac{R_{q,p}}{\Gamma_{q,p}^{\star}} }{\lambda + \sum_{l=1}^{L_c} \frac{R_{q,l}}{\Gamma_{q,l}^{\star}} }, \label{eq:def_Delta} \\
\chi_{p} &= \mu_{p} \left( \frac{1-\rho_x}{\hm_{p}} \aag_{\sfc}(\hchi_{p}) + \frac{\rho_x}{\hm_{p}} \aag_{\sfc}(\hm_{p}^2+\hchi_{p}) \right).
\end{align}
The parameters $\Gamma_{q,p}^{\star}$ and $\hchi_{p} = \sum_{q=1}^{L_r} \hchi_{q,p}$ are the solutions of the coupled equations
\begin{subequations} \label{eq:mQx_Pro1}
\begin{align}
\Gamma_{q,p}^{\star} &= \frac{1 -\Delta_{q,p}}{\chi_{p}}, \label{eq:mQx_Pro1a} \\
\hchi_{q,p} &=  \sum_{r=1}^{L_c} \left(\mse_{r}- \frac{\sigma_0^2\chi_{r}}{\lambda} \right)\Gamma'_{q,p,r} +
\frac{\mse_{p}}{\chi_{p}^{2}} - \frac{\sigma_0^2}{\lambda} \Gamma_{q,p}^{\star},
\end{align}
\end{subequations}
where
\begin{align}
\Gamma'_{q,p,r} = \left(\frac{1}{\nu_{q}} \frac{\Delta_{q,p}\Delta_{q,r}\Gamma_{q,p}^{\star}\Gamma_{q,r}^{\star}}{(1-2\Delta_{q,p})(1-2\Delta_{q,r})} \left( 1+ \sum_{l=1}^{L_c} \frac{1}{\nu_{q}} \frac{\Delta_{q,l}^2}{1-2\Delta_{q,l}}\right)^{-1}
- \frac{\Gamma_{q,r}^{\star2}}{1-2\Delta_{q,r}} \delta_{p,r} \right).
\end{align}
\end{Proposition}

\begin{proof}
See Appendix A.
\end{proof}

Note that except for $\{ m_p, Q_p\}$, the remaining parameters in Proposition \ref{Pro1} are arisen from the replica analysis and can be regarded auxiliary. The parameters $\{\Gamma_{q,p}^{\star}, \hchi_{q,p}\}$ have to be solved in (\ref{eq:mQx_Pro1}) for all $p,q$.%$p=1,\ldots,L_c$ and $q=1,\ldots,L_r$.

Proposition \ref{Pro1} provides not only a new finding but also a unified formula that embraces previous known results \cite{Tulino-13IT,Vehkapera-arXiv13}. For example, the MSE of LASSO under Type-A measurement matrix in \cite{Tulino-13IT} can be obtained if we set $L_c = L_r = 1$ and $ \mu_1 = 1$ in Propositions \ref{Pro1}. Clearly, by setting $\mu_1 < 1$, we are also able to further study the MSE of LASSO under Type-B measurement matrix. In the next section, we will discuss the MSEs of LASSO under Type-A and Type-B measurement matrices and we compare their performances and behaviors with those for random i.i.d.~Gaussian matrices.

Another existing result is related to the Type-C.1 measurement matrix in \cite{Vehkapera-arXiv13} in which $L_r = 1$ and $ \mu_p = 1$ for $p=1,\ldots,L_c$. In \cite{Vehkapera-arXiv13}, a Type-C.1 orthogonal matrix is referred to as the $T$-orthogonal matrix as the matrix is constructed by concatenating $T$ independent standard orthonormal matrices. Also, \cite{Vehkapera-arXiv13} only considered the real-valued setting, where the signal $\qx^0$, the measurements $\qy$, and the measurement matrix $\qA$ are all real-valued. In this case, the $\ell_1$-norm is defined as $\|\qx\|_1 \triangleq \sum_{n} |x_n|$, which is different from the complex $\ell_1$-norm (see footnote $1$). In the real-valued setting, the analytical MSE expression of LASSO in Proposition \ref{Pro1} also holds while $g_{\sfc}$ and $\aag_{\sfc}$ in (\ref{eq:defgc}) should be replaced by
\begin{subequations} \label{eq:defgr}
\begin{align}
g_{\sfr}(\zeta) &\triangleq - 2 \left( \sqrt{\frac{\zeta}{2 \pi}} e^{-\frac{1}{2\zeta}} - (1+\zeta) \sfQ\left(\frac{1}{\sqrt{\zeta}}\right) \right), \\
\aag_{\sfr}(\zeta) &\triangleq 2 \sfQ\left(\frac{1}{\sqrt{\zeta}}\right).
\end{align}
\end{subequations}

The difference between (\ref{eq:defgc}) and (\ref{eq:defgr}) is significant, and can be understood from (\ref{eq:RLS_EqScalSysModel}) by considering its real-valued counterpart. In the real-valued setting of (\ref{eq:RLS_EqScalSysModel}), $x_p^0$ and $y_p$ are real-valued, and the optimal solution becomes $\hat{x}_p = \left( |y_p| - \frac{1}{2}\right)_{+}/\hm_{p}$, which is quite different from its complex-valued counterpart in (\ref{eq:optOfxp}). This difference turns out to be reflected on $m_p$ and $Q_p$ and thus on $g_{\sfc}$ and $\aag_{\sfc}$ of (\ref{eq:mQ_Pro1}). In particular, the MSE of LASSO with the $T$-orthogonal matrix \cite{Vehkapera-arXiv13} can be perfectly recovered if we set $L_r = 1$, and those $ \mu_p = 1$ and $R_{1,p}=1$ for $p=1,\ldots,L_c$ in Proposition \ref{Pro1} and replace $g_{\sfc}$ and $\aag_{\sfc}$ by $g_{\sfr}$ and $\aag_{\sfr}$, respectively. Clearly, Proposition \ref{Pro1} provides a unified result that allows us to quantify the MSE of LASSO under a variety of measurement matrices. We will present detailed discussions in the next section.

\section*{IV. Discussions}
\subsection*{A. Type-B Orthogonal Matrix}
In this subsection, we aim to study the MSE of LASSO under Type-A and Type-B measurement matrices. In particular, we will compare their performances and behaviors with those for random i.i.d.~Gaussian matrices. To make comparison fair between different setups, all cases of the measurement matrices are normalized so that $\aang{\tr(\qA\qA^H)}_{\qA} = \oM$ (referred to as the power constraint of the measurement matrix). If the elements of $\qA$ are i.i.d.~Gaussian random variables with zero mean and variance $1/\oN$, then the power constraint of the measurement matrix is satisfied. We call this matrix the i.i.d.~Gaussian matrix. On the other hand, if $\qA$ is a Type-A matrix, the power constraint of the measurement matrix is naturally satisfied, and in fact, it satisfies the more stringent condition $\tr(\qA\qA^H) = \oM$. Meanwhile, in the Type-B setup, we set the gain factor $R_{1,1} = N/\oN = 1/\mu$ to satisfy this power constraint.

Since there is only one block, i.e., $L_c=L_r=1$, in the Type-A and Type-B setups, we omit the block index $(q,p)$ from all the concerned parameters hereafter, and Proposition \ref{Pro1} is anticipated to be greatly simplified. The MSE of LASSO under Type-B orthogonal measurement matrix is given as follows.

\begin{Corollary} \label{Cor1}
With the Type-B orthogonal measurement matrix, the MSE of LASSO is given by $\mse = \rho_x-2m+Q$, where $(m,Q)$ are same as those in (\ref{eq:mQ_Pro1}) while the block index $p$ is omitted. The parameters $(m,Q)$ are functions of $(\hm,\hchi)$ which can be obtained by solving the following set of equations
\begin{subequations} \label{eq:par1_TypeB}
\begin{align}
   \chi &= \mu \left( \frac{1-\rho_x}{\hm} \aag_{\sfc}(\hchi) + \frac{\rho_x}{\hm} \aag_{\sfc}(\hm^2+\hchi) \right), \\
   \hchi &= \mse \left( \Gamma'+ \frac{1}{\chi^{2}} \right) - \frac{\sigma_0^2}{\lambda} \left( \chi \Gamma' +\Gamma^{\star} \right),
\end{align}
\end{subequations}
with the following definitions
\begin{subequations} \label{eq:par2_TypeB}
\begin{align}
   \hm &\triangleq \frac{\lambda+R\chi + \sqrt{\left(\lambda-R\chi\right)^{2}-4\lambda\nu R\chi}}{2\lambda \chi}, \\
   \Gamma^{\star} &\triangleq \frac{\lambda-R\chi - \sqrt{\left(\lambda-R\chi\right)^{2}-4\lambda\nu R\chi}}{2\lambda \chi}, \\
   \Gamma' &\triangleq -\left(\frac{1-R^{-1}}{\Gamma^{\star 2}} + \frac{R^{-1}\lambda^2}{(\lambda\Gamma^{\star}+R)^2} + \frac{\nu R^2-R}{\Gamma^{\star 2}(\lambda\Gamma^{\star}+R)^2} \right)^{-1}.
\end{align}
\end{subequations}
\end{Corollary}

\begin{proof}
The above results can be obtained by substituting the corresponding parameters of Type-B setup, i.e., $L_c=L_r=1$, into Proposition \ref{Pro1}. In addition, using (\ref{eq:def_Delta}) and (\ref{eq:mQx_Pro1a}), we have eliminated $\Delta$.
\end{proof}

Let us first consider the Type-A setup, where we have $\nu=\alpha$ and $\mu=1$. If we set the gain factor $R = 1/\alpha$, then we have $\nu R = 1$. Then (\ref{eq:par2_TypeB}) can be further simplified in the form
\begin{subequations} \label{eq:par2_TypeA}
\begin{align}
   \hm &= \frac{\lambda+\alpha^{-1}\chi + \sqrt{\left(\lambda-\alpha^{-1}\chi\right)^{2}-4\lambda \chi}}{2\lambda \chi}, \\
   \Gamma^{\star} &= \frac{\lambda-\alpha^{-1}\chi - \sqrt{\left(\lambda-\alpha^{-1}\chi\right)^{2}-4\lambda \chi}}{2\lambda \chi}, \\
   \Gamma' &= -\left(\frac{1-\alpha}{\Gamma^{\star 2}} + \frac{\alpha\lambda^2}{(\lambda\Gamma^{\star}+\alpha^{-1})^2} \right)^{-1}.
\end{align}
\end{subequations}
Recall that in the real-valued setting, $g_{\sfc}$ and $\aag_{\sfc}$ in (\ref{eq:defgc}) should be replaced by $g_{\sfr}$ and $\aag_{\sfr}$ in (\ref{eq:defgr}). In this case, the above result gives exactly the same MSE result as reported in \cite[Example 2]{Vehkapera-arXiv13}. It should be noticed that the setting of ``$R = 1/\alpha$'' is used above to align the setting of \cite{Vehkapera-arXiv13}. According to the power constraint of the measurement matrix in this paper, we should set $R = 1$ rather than $R = 1/\alpha$.

Before proceeding, we present numerical experiments to verify our theoretical results. In the experiments, Type-A and Type-B matrices were generated from a
randomly scrambled $N\times N$ DFT matrix with $N=2^{15}$. The proximal gradient method (\ref{eq:proxGrad}) was used to solve LASSO and obtain the reconstruction
$\hat{\qx}$. The experimental average MSE was obtained by averaging over $10,000$ independent realizations. To form a measurement matrix, the selected column and
row sets at each realization were changed randomly. The experimental average MSEs of LASSO under different selecting rates $\mu$ are listed in Table
\ref{tab:difSelRate} in which the theoretical MSE estimates by Corollary \ref{Cor1} are also listed for comparison, with the parameters: $\alpha=0.5$,
$\lambda=0.1$, $\sigma_0^2 = 10^{-2}$, and $\rho_x=0.15$. In Table \ref{tab:difRegPar}, we fixed the selecting rate $\mu= 0.75$ and repeated the previous
experiment with different regularization parameters $\lambda$. Finally, in Table \ref{tab:difNoiseLevel}, we fixed the selecting rate $\mu = 0.75$ and
regularization parameter $\lambda=0.1$ and repeated the experiment with different noise levels $\sigma_0^2$. We see that for all the cases, the differences between
the two estimates are inappreciable. Therefore, Corollary \ref{Cor1} provides an excellent estimate of the MSE of LASSO in large systems.

\begin{table}
\caption{Comparison between experimental and theoretical MSEs of LASSO under different selecting rates $\mu$ for $\alpha=0.5$,  $\lambda=0.1$, $\sigma_0^2 =10^{-2}$, and $\rho_x=0.15$.}
\begin{center}
\begin{tabular}{c||ccccc}
\toprule
 $\mu$ & $1$ & $0.75$ & $0.5$ & $0.25$ & $0.1$ \\
%\midrule
\hline
 Theory (dB) &  $-19.11$ & $-18.72$ & $-18.37$ & $-18.06$ & $-17.88$ \\
%\cmidrule(r){1-6}
\hline
 Experiment (dB) & $-19.11$ & $-18.72$ & $-18.37$  & $-18.07$ & $-17.91$ \\
\bottomrule
\end{tabular}
\end{center}\label{tab:difSelRate}
\end{table}

\begin{table}
\caption{Comparison between experimental and theoretical MSEs of LASSO under different regularization parameter $\lambda$ for $\alpha=0.5$, $\mu=0.75$, $\sigma_0^2 = 10^{-2}$, and $\rho_x=0.15$.}
\begin{center}
\begin{tabular}{c||ccccc}
\toprule
 $\lambda$ & $0.1$ & $0.3$ & $0.5$ & $0.7$ & $0.9$ \\
%\midrule
\hline
 Theory (dB) &  $-18.72$ & $-16.39$ & $-13.72$ & $-11.91$ & $-10.70$ \\
%\cmidrule(r){1-6}
\hline
 Experiment (dB) & $-18.72$ & $-16.39$ & $-13.72$ & $-11.91$ & $-10.70$ \\
\bottomrule
\end{tabular}
\end{center}\label{tab:difRegPar}
\end{table}

\begin{table}
\caption{Comparison between experimental and theoretical MSEs of LASSO under different noise levels $\sigma_0^2$ for $\alpha=0.5$, $\mu=0.75$, $\lambda=0.1$, and $\rho_x=0.15$.}
\begin{center}
\begin{tabular}{c||ccccc}
\toprule
 $1/\sigma_0^2$ (dB) & $0$ & $10$ & $20$ & $30$ & $40$\\
%\midrule
\hline
Theory (dB) &  $0.013$ & $-10.45$ & $-21.09$ & $-28.05$ & $-29.04$ \\
%\cmidrule(r){1-6}
\hline
Experiment (dB) &  $0.013$ & $-10.45$ & $-21.10$ & $-28.06$ & $-29.05$  \\
\bottomrule
\end{tabular}
\end{center}\label{tab:difNoiseLevel}
\end{table}

Notice that unlike those works (e.g., \cite{Vehkapera-arXiv13}) employing {\tt CVX} \cite{cvx} for the LASSO problem, we used the proximal gradient method in
conjunction with the FFT operators, which allows us to deal with signal sizes as large as $10^5$ on a typical personal computer in about a few seconds. This
indicates that orthogonal matrices are highly relevant for large-scale compressive sensing applications and as a result the theoretical result based on the
assumption of $N \rightarrow \infty$ is useful. Even so, it would be essential to understand how far away the result based on infinitely large system is from that
for finite-sized systems. To this end, we plot in Figure \ref{fig:Convergences} the average MSEs of LASSO for different sizes of $N \times N$ DFT (or DCT) matrices
in which the measurement ratio is either $\alpha=0.5$ or $\alpha=0.35$, and the other parameters are fixed to be $\mu=0.75$, $\lambda=0.1$, and $\rho_x=1/8$.
Markers correspond to the experimental average MSEs for $N = 2^5, 2^6, \ldots, 2^{14}$, averaged over $100,000$ realizations of the LASSO problem under the measurement matrices constructed by the partial DFT and DCT matrices. Following the
methodology of \cite{Kabashima-12JSM,Vehkapera-arXiv13}, solid and dashed lines are plotted based on the experimental MSEs fitted with a quadratic function of $1/N$, which
provide the experimental estimates of the average MSEs. Extrapolation for $N\rightarrow \infty$ provides the estimates for the experimental MSE. Filled markers
represent the predictions obtained through Corollary \ref{Cor1}. We see that as $N \rightarrow \infty$, the experimental and the theoretical MSEs are identical as
expected. In addition, as the dimension of the system is above the order of $10^2$, the theoretical MSE provides a realizable prediction whatever $N \times N$
standard orthonormal matrices are adopted.

\begin{figure}
\begin{center}
\resizebox{4.5in}{!}{%
\includegraphics*{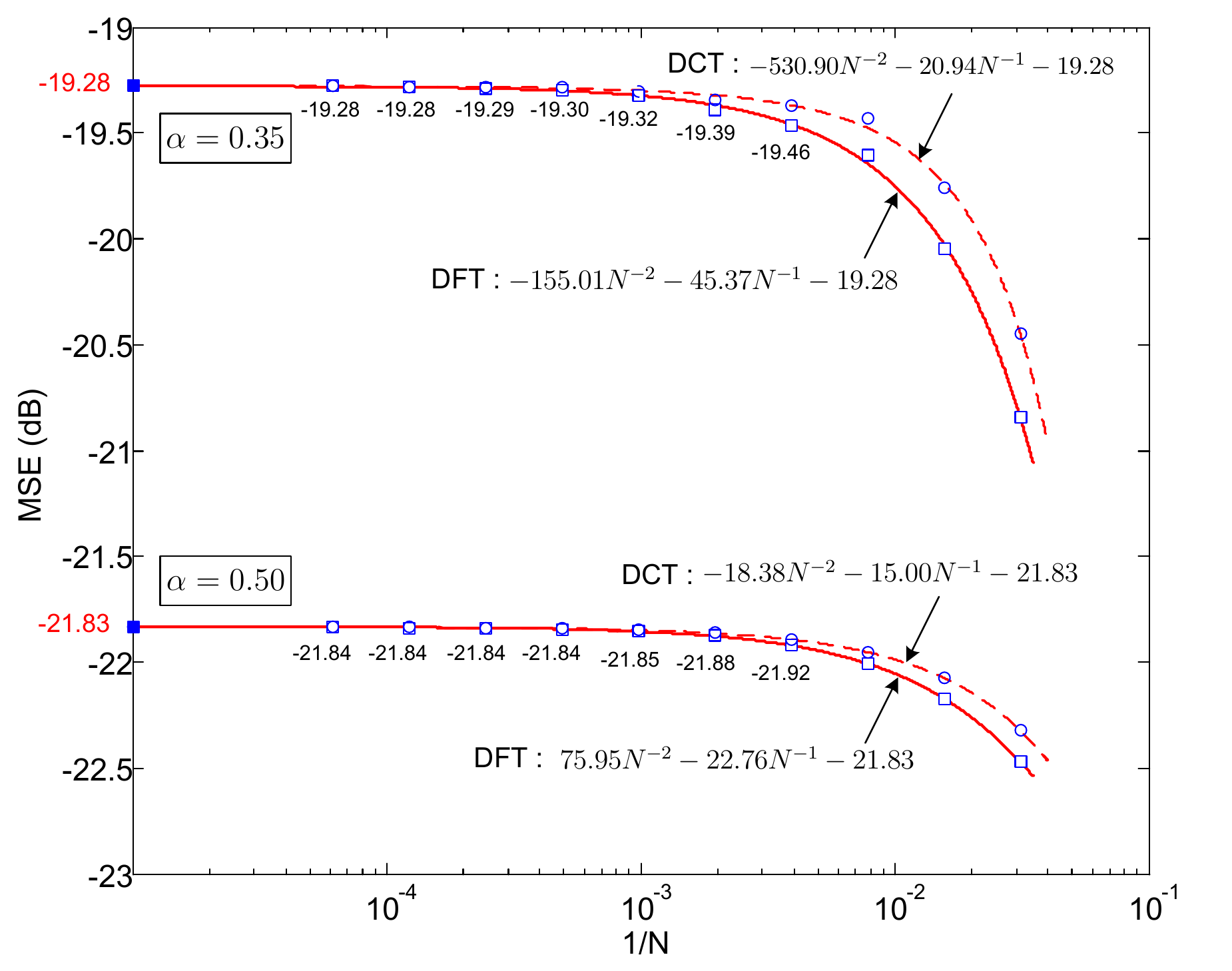} }%
\caption{Average MSE against the system dimension for $\alpha=0.35$ and $\alpha=0.5$. The other parameters are set as $\mu=0.75$, $\lambda=0.1$, and $\rho_x=1/8$. The markers correspond to the experimental average MSEs for $N = 2^5, 2^6, \ldots, 2^{14}$, averaged over $10^5$ realizations of the LASSO problem. Solid and dashed lines are plotted based on the experimental MSEs fitted with a quadratic function of $1/N$, which provided the experimental estimates of the average MSEs. Filled markers are the predictions using Corollary \ref{Cor1}.}\label{fig:Convergences}
\end{center}
\end{figure}

Because of the high accuracy and simplicity of the analytical result, we use the theoretical expression to study the behaviors of the MSE. In Figure
\ref{fig:iidVsDFT}, we compare the MSEs of LASSO for various regularization parameter $\lambda$ under different types of measurement matrices. The solid line,
dotted lines, and dashed line correspond to the MSEs under Type-A setup, Type-B setup, and the i.i.d.~Gaussian setup, respectively. The theoretical MSE of LASSO
under i.i.d.~Gaussian matrices is given by \cite[(133)]{Tulino-13IT} while those under Type-A and Type-B matrices are given by Corollary \ref{Cor1}. In fact, the
Type-A setup can be obtained by setting the column selection rate as $\mu = 1.00$ in the Type-B setup. Therefore, the column selection rates from bottom to top are
$\mu = 1.00$ (Type-A) and $\mu=0.75, 0.5, 0.25, 0.01$ (Type-B). In this experiment, the measurement ratio is fixed to be $\alpha=0.5$. Therefore, in order to
satisfy the fixed measurement ratio, we vary the row selection rate according to $\nu = \alpha \times \mu$ for each $\mu$. Figure \ref{fig:iidVsDFT} demonstrates
that Type-A setup has the best MSE performance compared to Type-B setups and is significantly better than the case of random i.i.d.~Gaussian matrices. In contrast
to Type-A matrices, the row-orthogonality in Type-B is no longer preserved if $\mu < 1$. The MSE performance of Type-B matrices degrades with decreasing the column
selection rate $\mu$ while they are at least as good as their random i.i.d.~Gaussian counterparts. In addition, in Figure \ref{fig:iidVsDFT}, markers show the
lowest MSE with respect to the regularization parameter $\lambda$. As can be seen, the optimal value of $\lambda$ depends on the matrix ensemble though not so
sensitive.

\begin{figure}
\begin{center}
\resizebox{4.5in}{!}{%
\includegraphics*{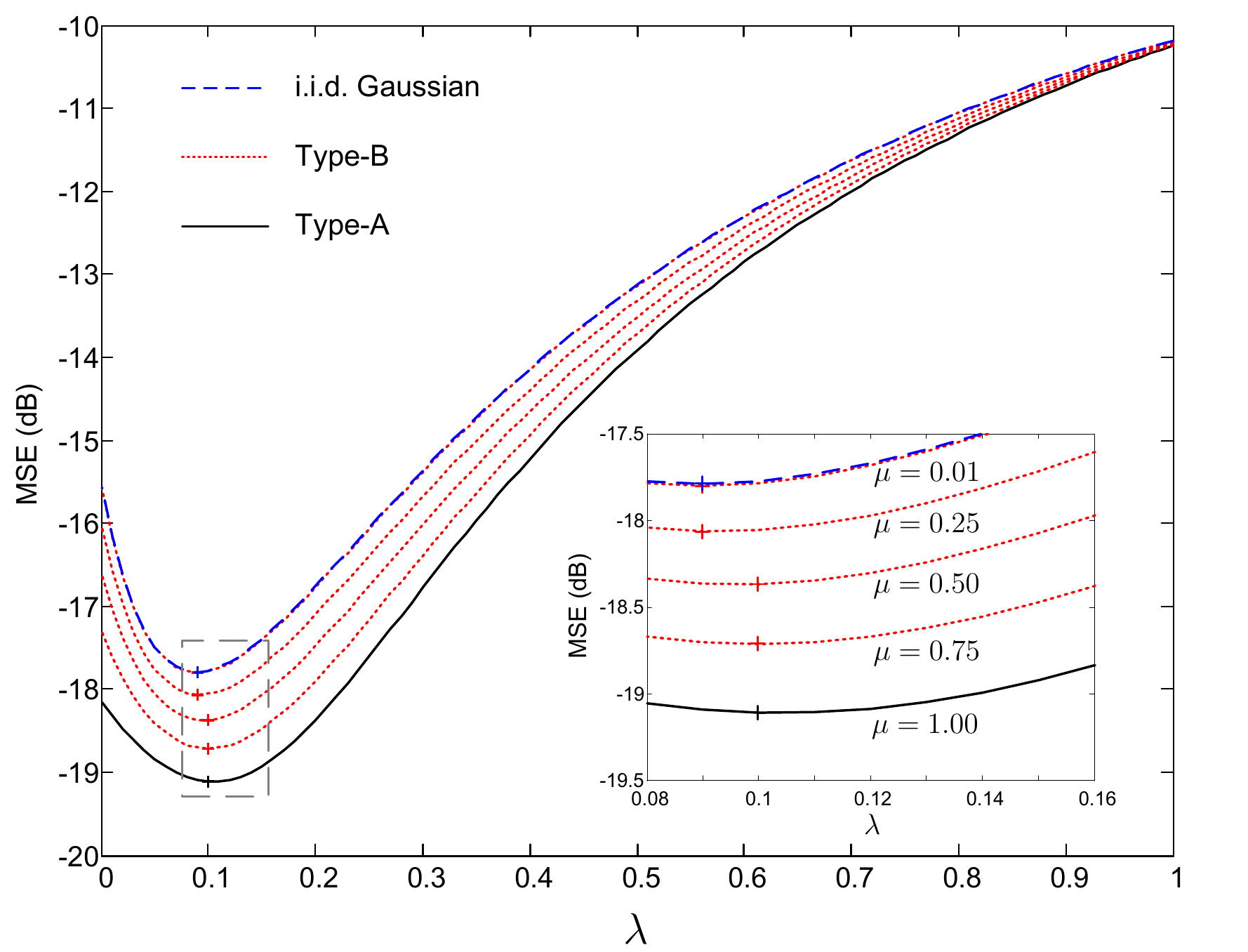} }%
\caption{Average MSE against the regularization parameter $\lambda$ for $\alpha=0.5$, $\sigma_0^2 = 10^{-2}$, and $\rho_x=0.15$. The dashed line is the MSE of the i.i.d.~Gaussian setup, the solid line is the MSE of Type-A setup, and the dotted lines are the MSE of Type-B setup with the selection rate $\mu = 0.75, 0.5, 0.25, 0.01$ (from bottom to top). Markers correspond to the lowest MSE with respect to the regularization parameter.}\label{fig:iidVsDFT}
\end{center}
\end{figure}

In \cite{Oymak-14ISIT}, it was argued that for \emph{noisy} measurements (\ref{eq:sysModel}), singular value distribution of the measurement matrix plays a key role in LASSO reconstruction, and measurement matrices with similar singular value characteristics exhibit similar MSE behavior. Together with this argument and Figure \ref{fig:iidVsDFT}, we may infer that the singular value distribution of Type-B matrix with \emph{small} column selection rate will be similar to that of the i.i.d.~Gaussian matrix. To demonstrate this aspect, the following theorem is useful.

\begin{Theorem} \label{Th1}
In the Type-B setup, the matrix $\qA$ is constructed by randomly selecting $\oM$ rows and $\oN$ columns from the standard orthonormal matrix multiplied by $\sqrt{R}$. Denote the row and column selection rate by $\nu = \oM/N$ and $\mu = \oN/N$, respectively. Then the empirical distribution of the $\min\{ \oM, \oN \}$ largest eigenvalues of $\qA\qA^H$ converges almost surely to
\begin{equation} \label{eq:HaarDensity}
f_{\sf Haar}(x) = \frac{\sqrt{\left( x-a_{-} \right)_{+} \left( a_{+}-x \right)_{+} }}{2 \pi x (R-x) (1-\max\{1-\nu,1-\mu \})} {\sf I}_{\{ a_{-}\leq x \leq a_{+} \}}
+ \frac{(\nu+\mu-1)_{+}}{1-\max\{1-\nu,1-\mu \}} \delta(x-R),
\end{equation}
where
\begin{subequations}
\begin{align}
    a_{+} &= R \left( \sqrt{(1-\mu)\nu} + \sqrt{(1-\nu)\mu} \right)^2, \\
    a_{-} &= R \left( \sqrt{(1-\mu)\nu} - \sqrt{(1-\nu)\mu} \right)^2.
\end{align}
\end{subequations}
\end{Theorem}
\begin{proof}
This theorem can be easily worked out from \cite[Theorem 3.1]{Farrell-11JFAA} by carefully substituting the setup of Type-B, which completes the proof.
\end{proof}

\begin{figure}
\begin{center}
\resizebox{5.5in}{!}{%
\includegraphics*{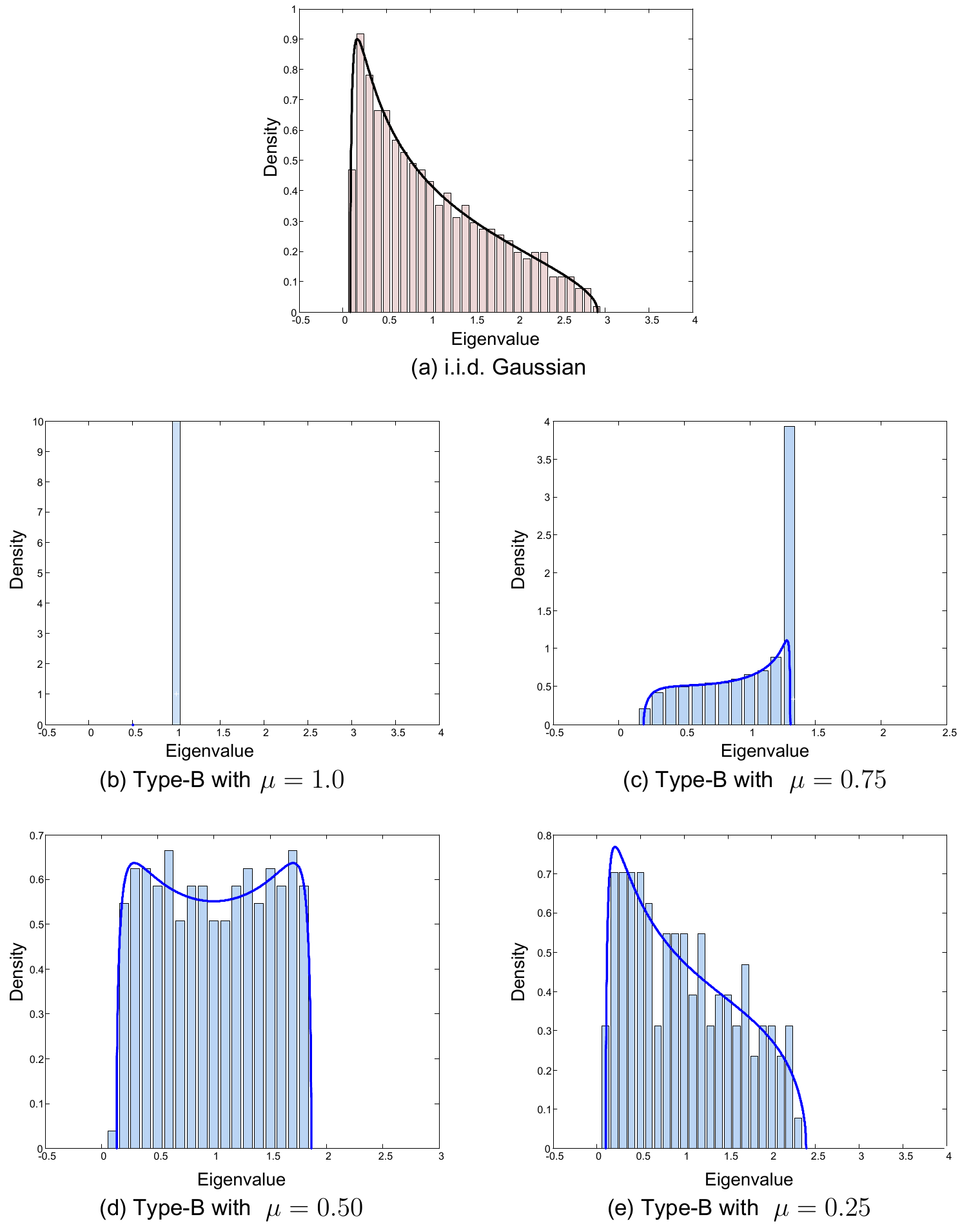} }%
\caption{Empirical eigenvalue distribution for one realization plotted against continuous part for asymptotic distribution. In (a), the i.i.d.~Gaussian matrix has dimensions $512 \times 1024$. In (b)--(e), the original DFT matrix has dimensions $1024 \times 1024$.}\label{fig:eig1}
\end{center}
\end{figure}

Before discussions, let us demonstrate the empirical distribution of the i.i.d.~Gaussian matrix. Recall (from the beginning of this subsection) that a matrix $\qA\in \bbC^{\oM \times \oN}$ is called the i.i.d.~Gaussian matrix if the elements of $\qA$ are i.i.d.~Gaussian random variables with zero mean and variance $1/\oN$. The empirical distribution of $\qA\qA^H$ converges  almost surely to the well-known Mar\v{c}enko-Pastur distribution \cite{Bai-10}:
\begin{equation} \label{eq:GaussianDensity}
f_{\sf Gaussian}(x) = \frac{\sqrt{\left( x-b_{-} \right)_{+} \left( b_{+}-x \right)_{+} }}{2 \pi x \alpha} {\sf I}_{\{ b_{-}\leq x \leq b_{+} \}}
 + (1-\alpha^{-1})_{+} \delta(x),
\end{equation}
where $b_{\pm} = (1 \pm \sqrt{\alpha})^2$ and $\alpha = \oM/\oN$.

In Figure \ref{fig:eig1}, we compare the eigenvalue distributions of the i.i.d.~Gaussian matrix and Type-B matrices using the similar setups in Figure \ref{fig:iidVsDFT}. The continuous part indicates the asymptotic distribution while the histogram represents the empirical eigenvalues of $\qA\qA^H$. Recall that the measurement ratio is fixed to $\alpha= 0.5$. In addition, to satisfy the power constraint of the measurement matrix, we have set the gain factor $R = 1/\mu$. We see in Figure \ref{fig:eig1}(b) that the eigenvalues of Type-A matrix (or Type-B matrix with $\mu = 1.0$) are all located at one as expected. Next, the eigenvalues are diverged from one if $\mu < 1$. For some values of $\mu$, e.g., $0.5 < \mu < 1$ and see Figure \ref{fig:eig1}(c), the peak of the probability density appears close to the largest eigenvalue $a_{+}$. For other values of $\mu$, e.g., $\mu < 0.5$ and see Figure \ref{fig:eig1}(e), the eigenvalues start to cluster near zero, and the probability density looks like that of the i.i.d.~Gaussian matrix.

In Figure \ref{fig:eig2}, we plot the asymptotic eigenvalue distributions of the i.i.d.~Gaussian matrix and Type-B matrices for small values of $\mu$. As expected, the asymptotic eigenvalue distributions of the i.i.d.~Gaussian matrix and Type-B matrix with $\mu = 0.01$ are almost indistinguishable. In fact, this observation can be obtained from Theorem \ref{Th1}. To do so, in particular, we substitute $\mu, \nu \rightarrow 0$ with $\nu/\mu = \alpha \in (0,~1]$ and $R=1/\mu$ into (\ref{eq:HaarDensity}). Then the second term of (\ref{eq:HaarDensity}) is removed because $\mu + \nu < 1$. In addition, we have $(R-x) (1-\max\{1-\nu,1-\mu \}) \rightarrow \alpha$ and $a_{\pm} \rightarrow (1 \pm \sqrt{\alpha})^2$. Clearly, $f_{\sf Haar}(x)$ converges to $f_{\sf Gaussian}(x)$. We thus bridge a way to meet the i.i.d.~Gaussian matrix from a Haar measure of random matrix. Our result further supports the argument of \cite{Oymak-14ISIT} that for the noisy setup (\ref{eq:sysModel}), measurement matrices with similar singular value characteristics exhibit similar MSE behavior in LASSO formulation. From Figures \ref{fig:iidVsDFT} and \ref{fig:eig1}, one can realize that the divergence of the eigenvalues also results in performance degeneration.

\begin{figure}
\begin{center}
\resizebox{4.5in}{!}{%
\includegraphics*{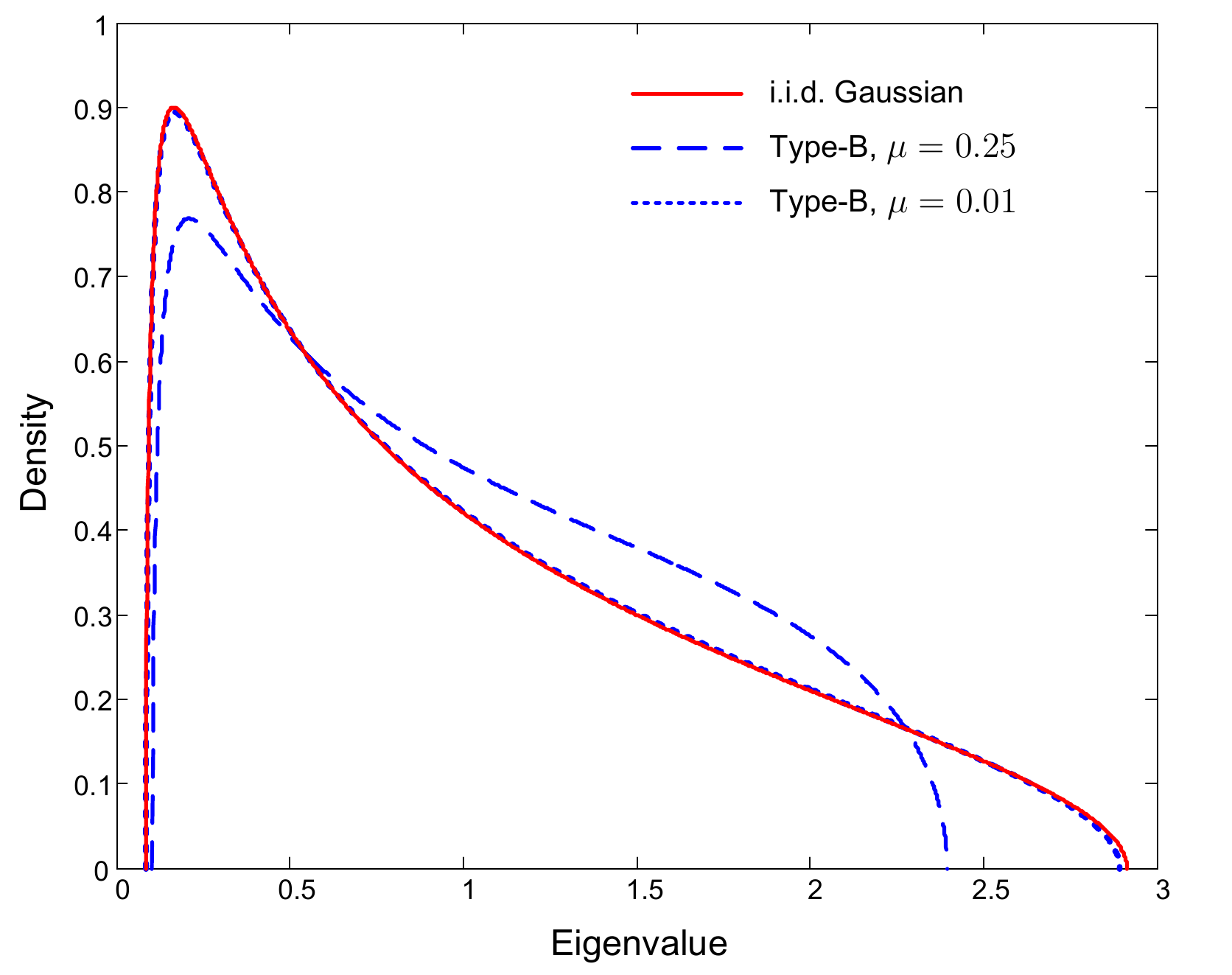} }%
\caption{The Mar\v{c}enko-Pastur distribution and the asymptotic eigenvalue distribution of Type-B for $\mu=0.25, 0.01$.}\label{fig:eig2}
\end{center}
\end{figure}

\subsection*{B. Type-C Matrix}
Next, we study the MSE of LASSO under Type-C measurement matrices. Recall that Type-C matrices are attractive due to implementation considerations \cite{Do-08ICASSP,Fowler-12FtSP}. For example, to recover a $0.75 \times 4096 = 3072$ sparse signals with $0.5 \times 4096 =2048$ measurements by using the DFT operators, we have at least four approaches: 1) Type-B measurement matrix by selecting $3072$ columns and $2048$ rows from the $4096 \times 4096$ DFT matrix; 2) Type-C.2 measurement matrix concatenated by two matrices with $\qA_{1,1},\qA_{1,2} \in \bbC^{2048 \times 1536}$. The two matrices are taken from partial scrambled $2048 \times 2048$ DFT matrices with the column selection rate $\mu = 0.75$; 3) Type-C.3 measurement matrix concatenated by two matrices with $\qA_{1,1} \in \bbC^{2048 \times 2048}$ and $\qA_{1,2} \in \bbC^{2048 \times 1024}$. The two matrices are taken from randomly scrambled $2048 \times 2048$ DFT matrices, and the additional column selection with rate $\mu_{1,2} = 0.5$ is used to get $\qA_{1,2}$; 4) Type-C.4 measurement matrix concatenated by six randomly scrambled $1024 \times 1024$ DFT matrices, namely $\qA_{1,1}, \qA_{1,2}, \qA_{1,3}, \qA_{2,1}, \qA_{2,2}, \qA_{2,3} \in \bbC^{1024 \times 1024}$. In contrast to the Type-B setup, the implementations of Type-C.2, Type-C.3, and Type-C.4 setups can exploit parallelism or distributed computation, wherein the Type-C.4 setup has the best structure for parallel  computations. Therefore, a naturally arising question is how their MSE performances are effected among the different measurement matrices.

\begin{figure}
\begin{center}
\resizebox{4.5in}{!}{%
\includegraphics*{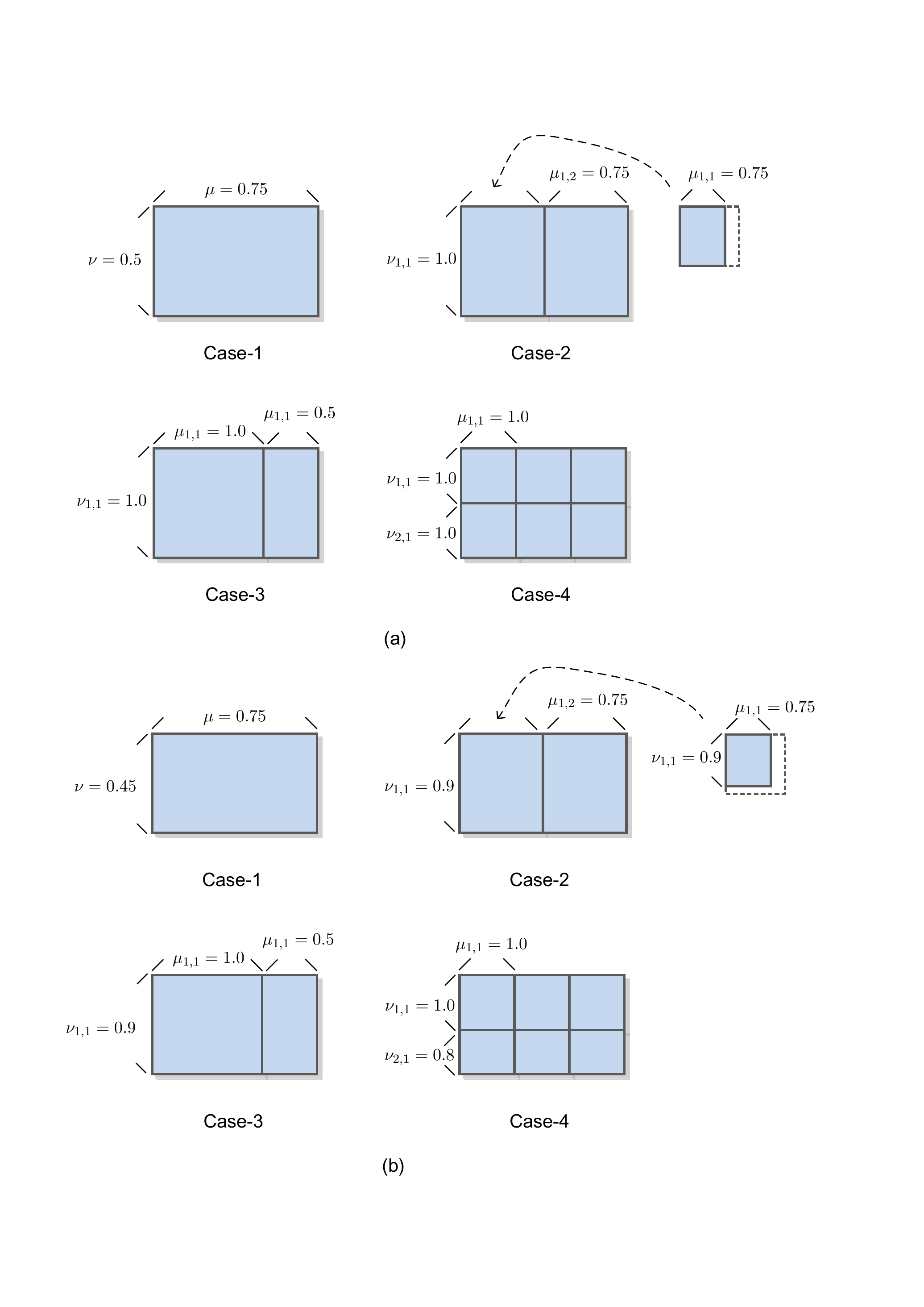} }%
\caption{The four examples of constructing measurement matrices with (a) $\oM = 0.5 \times 2^{15}$ and $\oN = 0.75 \times 2^{15}$, and with (b)  $\oM = 0.45 \times 2^{15}$ and $\oN = 0.75 \times 2^{15}$. Each block is obtained from an independent scrambled DFT matrix. The original DFT matrices of the four cases have dimensions $N=2^{15}$, $N=2^{14}$, $N=2^{14}$, and $N=2^{13}$, respectively.}\label{fig:FourCases}
\end{center}
\end{figure}

We also conducted extensive numerical experiments to verify the theoretical results in Proposition \ref{Pro1}. In the experiments, we use the example of the four cases mentioned above while enlarging the dimensions proportionally so that $\oM = 0.5 \times 2^{15}$ and $\oN = 0.75 \times 2^{15}$. For clarity, we depict the four cases in Figure \ref{fig:FourCases}(a). The experimental average MSEs of LASSO under the four measurement matrices are listed in Table \ref{tab:fourCases} in which the theoretical MSE estimates by Proposition \ref{Pro1} are also listed for comparison. In the same table, we repeat the previous experiment but using different measurement rate with $\oM = 0.45 \times 2^{15}$ and $\oN = 0.75 \times 2^{15}$. The corresponding four cases are depicted in Figure \ref{fig:FourCases}(b). As can be seen, for all cases in the tables, the differences between the two estimates are only evident in the last digits. Therefore, we confirm that Proposition \ref{Pro1} provides an excellent estimate of the MSE of LASSO in large systems.

\begin{table}
\caption{Comparison between the experimental and theoretical MSEs of LASSO under the four measurement matrices for $\lambda=0.1$, $\sigma_0^2 = 10^{-2}$, and $\rho_x=0.15$.}
\begin{center}
\begin{tabular}{l|l}
\toprule
 &
\begin{tabular}{lllll}
  \hspace{2.80cm} & \,Case-1 & \,Case-2 & ~Case-3 & ~Case-4
\end{tabular}  \\
\bottomrule
%\hline
 (a) $\oM/\oN = 0.5/0.75$ &
\begin{tabular}{lcccc}
 Theory (dB) & $-21.09$ & $-21.09$ & $-21.10$  & $-20.72$  \\
%\cmidrule(r){2-5}
\hline
 Experiment (dB) &  $-21.09$ & $-21.09$ & $-21.11$ & $-20.72$
\end{tabular}
\\
\bottomrule
 \toprule
 (b) $\oM/\oN = 0.45/0.75$ &
\begin{tabular}{lcccc}
 Theory (dB) &  $-20.32$ & $-20.32$ & $-20.36$ & $-19.96$ \\
%\cmidrule(r){2-5}
\hline
 Experiment (dB) & $-20.32$ & $-20.32$ & $-20.36$  & $-19.96$
\end{tabular} \\
\bottomrule
\end{tabular}

\end{center}\label{tab:fourCases}
\end{table}

Next, we use the theoretical expression to examine the behaviors of MSEs under Type-C measurement matrices. In Figure \ref{fig:MSE-FourCases}, we compare the MSEs of LASSO as a function of the regularization parameter $\lambda$ for the four cases depicted in Figure \ref{fig:FourCases}. The MSEs for Type-A and the i.i.d.~Gaussian counterparts are also plotted as references. As can be seen, Type-A setup always gives the best MSE result while the i.i.d.~Gaussian setup yields the worst MSE result. However, Type-A setup would not be always useful if the corresponding size of the FFT operators is not available in some DSP chips. Also, we see that Case-1 and Case-2 always have the same MSE behaviors. This finding motivates us to get the following observation that can meet the same performance of Type-B matrix via concatenating orthonormal bases.

\begin{figure}
\begin{center}
\resizebox{7.0in}{!}{%
\includegraphics*{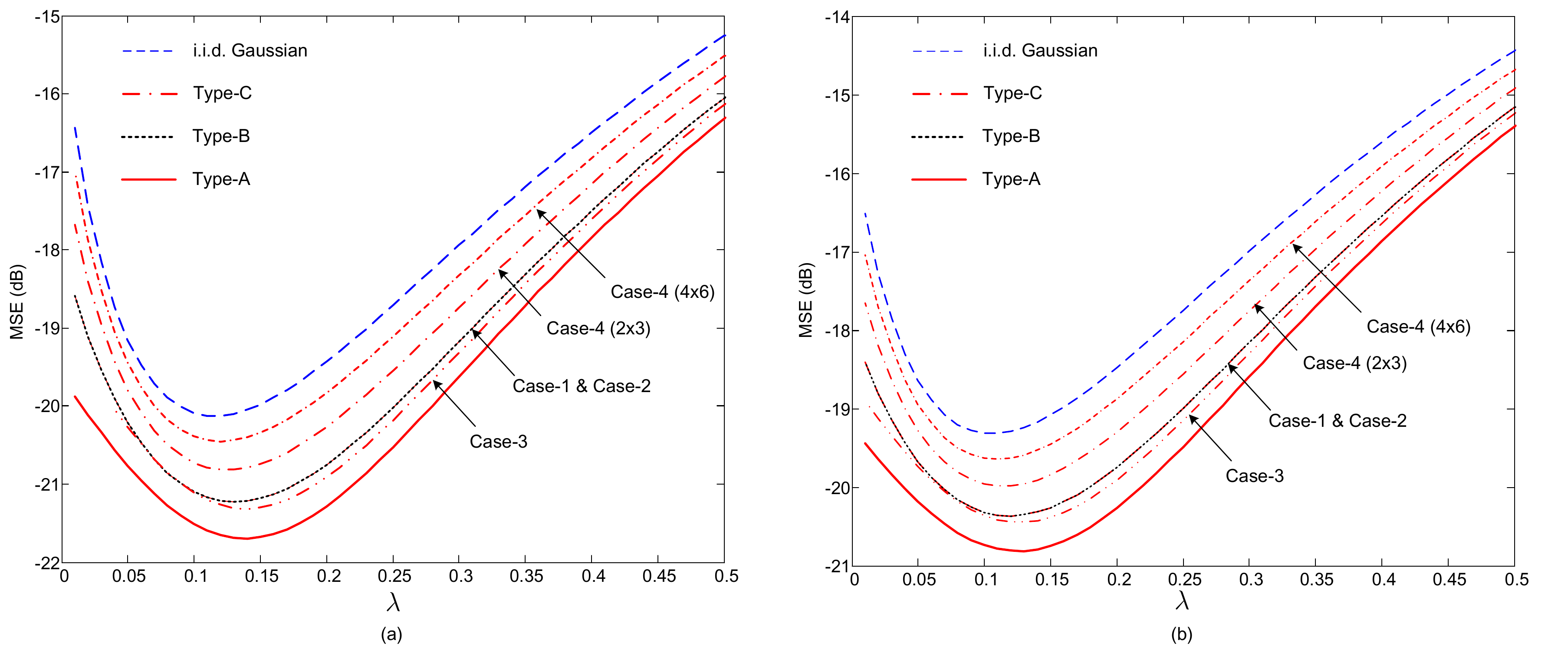} }%
\caption{Average MSE against the regularization parameter $\lambda$ for $\sigma_0^2 = 10^{-2}$ and $\rho_x=0.15$, and that (a) $\alpha = 0.5/0.75$ and (b) $\alpha = 0.45/0.75$. The dashed line is the MSE of the i.i.d.~Gaussian setup, the solid line is the MSE of Type-A setup, and the dotted lines are the MSE of Type-B setup with the selection rate $\mu = 0.75, 0.5, 0.25, 0.01$ (from bottom to top). Markers correspond to the lowest MSE with respect to the regularization parameter.}\label{fig:MSE-FourCases}
\end{center}
\end{figure}

\begin{Observation} \label{obs1}
Consider Type-B measurement matrix with the column and row selection rates $\mu$ and $\nu$. The MSE of LASSO under this measurement matrix is identical to that under the horizontal concatenation of $L_c$ matrices where each matrix is from a partial orthonormal matrix with the column and row selection rates $\mu$ and $L_c\nu$. For a meaningful construction, $L_c$ should be subjected to $L_c\nu \leq 1$.
\end{Observation}

To see an application of this observation, let us take two examples. First, consider Type-B measurement matrix with $\mu = 1.0$ and $\nu = 0.25$. Applying Observation \ref{obs1}, we have that the MSE of LASSO under the row-orthonormal measurement matrix is identical to that under the measurement matrix of $\qA =[\qA_{1,1} \,\qA_{1,2} \,\qA_{1,3} \,\qA_{1,4}]$ with each $\qA_{1,p}$ being a square orthogonal matrix. This argument was also revealed by \cite{Vehkapera-arXiv13}. It is noted that the columns of each $\qA_{1,p}$ are orthogonal so that there is no interference within each square orthogonal matrix. The interference resulting from the other sub-block of the measurement matrix will degenerate the MSE performance. Therefore, we can infer that the matrix constructed by concatenating many square orthonormal matrices should lose its advantage over the i.i.d.~Gaussian matrix. In other words, if the measurement ratio is small, i.e., $\oM \ll \oN$, the MSEs of LASSO under the row-orthonormal measurement matrix and the i.i.d.~Gaussian matrix should be comparable. This inference also seems reasonable from the aspect of eigenvalue spectrum \cite{Oymak-14ISIT} that: when $\oM \ll \oN$, an i.i.d.~Gaussian matrix has approximately orthogonal rows and it behaves similar to a row-orthonormal matrix.

Next, let us consider another example that Type-B measurement matrix is with $\mu = 0.6$ and $\nu = 0.3$. With Observation \ref{obs1}, we have that the MSE of LASSO under this measurement matrix is identical to that under the measurement matrix of $\qA = [\qA_{1,1} \,\qA_{1,2} \,\qA_{1,3}]$ with each $\qA_{1,p}$ being a partial orthogonal matrix with $\mu_{1,p} = 0.6$ and $\nu_{1,p} = 0.9$. In this case, the columns of each $\qA_{1,p}$ are not orthogonal any more but \emph{nearly} orthogonal. Therefore, we can expect some performance degeneration under this measurement matrix. Finally, it is clear that Case-1 and Case-2 in Figure \ref{fig:MSE-FourCases} have the same MSE behaviors but Case-2 has a better structure in the parallel computation and less requirement in the size of the FFT operator.

From Figure \ref{fig:MSE-FourCases}, we also observe that the measurement matrix constructed by vertical \emph{and} horizontal concatenation of several blocks,
i.e., Case-4, has the worst performance among the structurally orthogonal matrices. As a matter of fact, if we continue to increase the number of concatenation
blocks, then their MSE performances will degrade accordingly. However, in any cases, they are at least as good as their random i.i.d.~Gaussian counterparts. This
observation hence provides us another way to meet the random i.i.d.~Gaussian matrix via vertically and horizontally concatenating orthonormal bases.

Finally, comparing the four cases in Figure \ref{fig:MSE-FourCases}, we notice that if Type-A matrix is not available, Case-3 provides the best MSE result. This
observation together with the previous experiments indicate that to construct a measurement matrix aiming for a good MSE performance in LASSO formulation, one
should follow the example of Case-3. That is to say, first try to use a row-orthogonal matrix that can best fit the dimension of the measurement matrix and then
horizontally concatenate the remaining part.

\section*{V. Conclusion}
We have investigated the MSE performance of estimating a sparse vector through an undersampled set of noisy linear transformations when the measurement matrix is constructed by concatenating several randomly chosen orthonormal bases and LASSO formulation is adopted. Using the replica method in conjunction with some novel matrix integration results, we derived the theoretical MSE result. Extensive numerical experiments have illustrated excellent agreement with the theoretical result. Our numerical results also revealed the fact that the structurally orthogonal matrices are at least as well performed as the i.i.d.~Gaussian matrices. In particular, we have made the following observations:
\begin{itemize}
\item Type-A matrices (or row-orthogonal matrices) have the best MSE performance out of all the other types of structurally orthogonal matrices and is significantly better than the i.i.d.~Gaussian matrices.
\item The advantage of the row-orthogonal matrix over the i.i.d.~Gaussian matrix is still preserved even when a random set of columns is removed (which leads to a Type-B matrix). When increasing the number of the removed columns, the MSE of LASSO degenerates to the case of the i.i.d.~Gaussian matrices. In particular, we have shown that the asymptotic eigenvalue distribution of Type-B matrix with small column selection rate converges to that of the i.i.d.~Gaussian matrix.
\item In addition, a measurement matrix obtained by orthogonal matrix constructions has fast computation and facilitates parallel processing. For this purpose, we have provided a technique to meet the same performance of Type-B matrix via horizontally concatenating orthogonal bases. Our argument is more systematic than \cite{Vehkapera-arXiv13} and leads to much wider applications.
\item On the other hand, we have shown that the measurement matrix constructed by vertical concatenation of blocks usually gets the worst performance compared to the horizontal concatenation. However, they are at least as good as their random i.i.d.~Gaussian counterparts.
\end{itemize}
As a consequence, we conclude that in addition to the ease of implementation, the structurally orthogonal matrices are preferred for practical use in terms of their good estimation performance.

It was reported that orthogonal measurement matrices also enhance the signal reconstruction threshold in the noisy setups when the optimal Bayesian recovery is
used \cite{Wen-14ArXiv}. Promising future studies include performance evaluation under the optimal Bayesian recovery and development of recovery algorithms
suitable for the structurally orthogonal matrices \cite{Donoho-09PNAS,Rangan-10ArXiv,Krzakala-12JSM,Cakmak-14ArXiv,Ma-15SPL}.

\section*{\sc Appendix A: Proof of Proposition \ref{Pro1}}
First, recall that we have rewritten the average free entropy $\Phi$ in (\ref{eq:avgFreeEn}) by using the replica identity. Within the replica method, it is assumed that the limits of $\beta, N \rightarrow \infty$ and $\tau\rightarrow 0$ can be exchanged. We therefore write
\begin{equation}\label{eq:LimF_AppA}
\Phi = \lim_{\tau\rightarrow 0}\frac{\partial}{\partial \tau} \lim_{\beta, N \rightarrow \infty}\frac{1}{\beta N} \log\aang{\sfZ_{\beta}^{\tau}(\qy,\qA)}_{\qy,\qA}.
\end{equation}
We first evaluate $\aang{\sfZ_{\beta}^{\tau}(\qy,\qA)}_{\qy,\qA}$ for an integer-valued $\tau$, and then generalize it for any positive real number $\tau$. In particular, given the partition function of (\ref{eq:partFun_B}), we obtain
\begin{equation} \label{eq:ReplicPartFun_B}
 \sfZ_{\beta}^{\tau}(\qy,\qA)
 = \int d\qx^{(1)} \cdots d \qx^{(\tau)} \left( \prod_{a=1}^{\tau} e^{-\beta \|\qx^{(a)}\|_1} \right)
 e^{-\sum_{a=1}^{\tau}\frac{1}{\sigma_{a}^2}\left\| \qy-\qA \qx^{(a)} \right\|^2}
\end{equation}
with $\sigma_{a}^2 = \lambda/\beta \equiv \sigma^2$. Using the $\tau$-th moment of the partition function and $P(\qy|\qx^{0})$ in (\ref{eq:postDis_sys}), we have
\begin{equation} \label{eq:ReplicPartFun_B}
 \aang{\sfZ_{\beta}^{\tau}(\qy,\qA)}_{\qy,\qA}
 = \aangBig{ \frac{1}{(\pi\sigma_0^2)^{\oM}} \int \rmd \qy
 e^{-\sum_{a=0}^{\tau}\frac{1}{\sigma_{a}^2}\left\| \qy-\qA \qx^{(a)} \right\|^2}
 }_{\qA,\qX},
\end{equation}
where $\qx^{(a)} = \vect([\qx_1^{(a)} \ldots \qx_{L_c}^{(a)}])$ with $\qx_{p}^{(a)} $ being the $a$-th replica signal vector of $\qx_{p}$, and $\qX\triangleq
\{\qX_{p}, \forall p \}$ with $\qX_{p} \triangleq [ \qx_{p}^{(0)} \, \qx_{p}^{(1)} \cdots \, \qx_{p}^{(\tau)} ]$. The equality of (\ref{eq:ReplicPartFun_B})
follows from the fact that $\qx^{(a)}$ is a random vector taken from the input distribution $P_0(\qx)$ in (\ref{eq:px0}) if $a=0$ and $P_\beta(\qx) = e^{- \beta
\|\qx\|_1 }$ otherwise, and $\sigma_{a}^2 = \sigma_0^2$ if $a=0$ and $\sigma_{a}^2 = \sigma^2$ otherwise.

Before proceeding, we introduce the following preprocessing to deal with the cases in which $\qA_{q,p}$ is a randomly sampled orthogonal matrix (or, deleting
row/columns independently). In particular, we find that it is convenient to work with the enlarged orthogonal matrix $\tqA_{q,p} \in \bbC^{N \times N}$ with rows
and columns setting to zero rather removed \cite{Farrell-11JFAA}. For clarity, we use the following definition.

\begin{Definition} \cite{Farrell-11JFAA}
A square matrix is called a diagonal projection matrix if its off-diagonal entries are all zeros and its diagonal entries are zeros or ones.
\end{Definition}

Let $\qR_{q,p}$ and $\qT_{q,p}$ be $N \times N$ diagonal projection matrices, where the numbers of nonzero diagonal elements of $\qR_{q,p}$ and $\qT_{q,p}$ are $M_q$ and $N_p$, respectively. Therefore, we characterize each block by
\begin{equation} \label{eq:enlargeA}
 \tqA_{q,p} = \qR_{q,p}^{\frac{1}{2}} \qW_{q,p} \qT_{q,p}^{\frac{1}{2}} \in \bbC^{N \times N},
\end{equation}
where $\qW_{q,p}$ is an $N \times N$ standard orthonormal matrix. Since $\{ \qW_{q,p} \}$ are independent standard orthonormal matrices, the positions of nonzero elements of the diagonal projection matrices are irrelevant. For the sake of simplicity, we assume that all the diagonal entries $1$ of $\qR_{q,p}$ and $\qT_{q,p}$ appear first, i.e.,
\begin{equation} \label{eq:def_RT}
 \qR_{q,p} = \left[ \begin{array}{cc} \qI_{M_q} & \qzero \\ \qzero & \qzero  \end{array}\right]~~\mbox{and}~~
 \qT_{q,p} = \left[ \begin{array}{cc} \qI_{N_p} & \qzero \\ \qzero & \qzero  \end{array}\right], ~\forall p, q.
\end{equation}
Recall the Type-C matrix in Section II that the standard orthonormal matrix has been multiplied by $\sqrt{R_{q,p}}$. The gain factor $R_{q,p}$ can be included in $\qR_{q,p}$ via a scaling factor. For notational convenience, here, we do not use the expression of $R_{q,p}\qR_{q,p}$ but absorb $R_{q,p}$ into $\qR_{q,p}$. Also, we enlarge $\qx_p$ and $\qy_q$ to be $N$-dimensional vectors by zero padding. As a consequence, the input-output relationship of (\ref{eq:sysModel}) can be equivalently expressed as
\begin{equation}\label{eq:EnlargedSysModel}
 \underbrace{\left[\begin{array}{c}| \\\tqy_q\\ | \end{array}\right]}_{\triangleq\tqy}
 = \underbrace{\left[\begin{array}{ccc} & | & \\ - & \tqA_{q,p} & - \\ & | & \end{array}\right]}_{\triangleq\tqA}
    \underbrace{\left[\begin{array}{c} | \\ \qx_q \\ | \end{array}\right]}_{\triangleq\tqx}
 + \sigma_0 \underbrace{\left[\begin{array}{c} | \\\qw_q \\ | \end{array}\right]}_{\triangleq\tqw}.
\end{equation}
Notice that all the following derivations are based on the enlarged system (\ref{eq:EnlargedSysModel}). Therefore, by abuse of notation, we continue to write $\qx_p$, $\qy_q$, $\qA_{q,p}$, $\qx$, $\qy$, and $\qA$ for $\tqx_p \in \bbC^{N}$, $\tqy_q \in \bbC^{N}$,  $\tqA_{q,p} \in \bbC^{N \times N}$, $\tqx \in \bbC^{N L_c}$, $\tqy \in \bbC^{N L_r}$, and $\tqA \in \bbC^{N L_r \times N L_c}$, respectively.

Next, we introduce a random vector per block
\begin{equation} \label{eq:GaussAsum1}
\qv_{q,p}^{(a)}\triangleq \qT_{q,p}^{\frac{1}{2}} \qx_{p}^{(a)} \in \bbC^{N} ,
~\mbox{for}~ a=0,1,\ldots,\tau.
\end{equation}
The covariance of $\qv_{q,p}^{(a)}$ and $\qv_{q,p}^{(b)}$ is a $(\tau+1)\times(\tau+1)$ Hermitian $\qQ_{q,p}$ with entries given by
\begin{equation} \label{eq:defQ}
\left(\qv_{q,p}^{(a)}\right)^H \qv_{q,p}^{(b)}
= \left(\qx_{p}^{(a)}\right)^H\qT_{q,p} \left(\qx_{p}^{(b)}\right)
\triangleq N [\qQ_{q,p}]_{a,b} .
\end{equation}
For ease of exposition, we further write $\qV\triangleq\{ \qv_{q,p}^{(a)}, \forall a,b,k \}$, $\qW\triangleq\{ \qW_{q,p}, \forall q, p \}$, and
$\qQ\triangleq\{\qQ_{q,p}, \forall q, p\}$.

Now, we return to the calculation of (\ref{eq:ReplicPartFun_B}). In (\ref{eq:ReplicPartFun_B}), the expectations introduce iterations between $\qx$ and $\qA$. However, the resulting iterations depend only on the covariance as those shown in (\ref{eq:defQ}). Therefore, it is useful to separate the expectation over $\qX$ into an expectation over all possible covariance $\qQ$ and all possible $\qX$ configurations with respect to a prescribed set of $\qQ$ by introducing a $\delta$-function. As a result, (\ref{eq:ReplicPartFun_B}) can be rewritten as
\begin{equation}\label{eq:sf_E2}
 \aang{\sfZ_{\beta}^{\tau}(\qy,\qA)}_{\qy,\qA}=\aangBig{ e^{\beta N{\cal G}^{(\tau)}(\qQ)} }_{\qX}
 =\int e^{\beta N{\cal G}^{(\tau)}(\qQ)}\mu^{(\tau)}(\rmd\qQ),
\end{equation}
where
\begin{equation}\label{eq:G1}
{\cal G}^{(\tau)}(\qQ)\triangleq \frac{1}{\beta N} \log \aangBigg{ \frac{1}{(\pi\sigma_0^2)^\oM} \prod_{q=1}^{L_r} \int \rmd\qy_{q}e^{-\sum_{a=0}^{\tau}\frac{1}{\sigma_{a}^2}
\left\|\qy_{q} - \sum_{p=1}^{L_c}\qR_{q,p}^{\frac{1}{2}} \qW_{q,p}\qv_{q,p}^{(a)} \right\|^2} }_{\qW},
\end{equation}
and
\begin{equation}\label{eq:Measure_mu}
\mu^{(\tau)}(\rmd\qQ)\triangleq \aangBig{ \prod_{q=1}^{L_r}\prod_{p=1}^{L_c}\prod_{0\leq a \leq b}^{\tau}\delta\left(
\qx_{p}^{(a)H}\qT_{q,p}\qx_{p}^{(b)}-N [\qQ_{q,p}]_{a,b}\right) }_{\qX} \rmd\qQ.
\end{equation}
Next, we focus on the calculations of (\ref{eq:G1}) and (\ref{eq:Measure_mu}), respectively.

Let us first consider (\ref{eq:G1}). Integrating over $\qy_{q}$'s in (\ref{eq:G1}) by applying Lemma \ref{lemma_GaussianIntegr} yields
\begin{equation}\label{eq:G3}
{\cal G}^{(\tau)}(\qQ)= \frac{1}{\beta N} \log \aangBigg{ \prod_{q=1}^{L_r}e^{-\tr\left(\left(\sum_{p=1}^{L_c}\qR_{q,p}^{\frac{1}{2}} \qW_{q,p}\qV_{q,p}\right)
{\bf\Sigma}\left(\sum_{p=1}^{L_c}\qR_{q,p}^{\frac{1}{2}} \qW_{q,p}\qV_{q,p}\right)^H\right)} }_{\qW} - \frac{\nu}{\beta} \log\left( 1+\tau\frac{\sigma_0^2}{\sigma^2} \right),
\end{equation}
where
\begin{equation}
\qSigma\triangleq \frac{1}{\sigma^2(\sigma^2+\tau\sigma_0^2)}
\left[
\begin{array}{cc}
\tau\sigma^2 & -\sigma^2 \qone_{\tau}^T \\
-\sigma^2 \qone_{\tau} & (\sigma^2+\tau\sigma_0^2)\qI_{\tau} - \sigma_0^2\qone_{\tau} \qone_{\tau}^T
\end{array}
\right].
\end{equation}

Next, we consider (\ref{eq:Measure_mu}). Through the inverse Laplace transform of $\delta$-function, we can show that
\begin{equation} \label{eq:Measure_muInN}
\mu^{(\tau)}(\qQ) = e^{-\beta N\calR^{(\tau)}(\qQ)+\calO(1)},
\end{equation}
where $\calR^{(\tau)}(\qQ)$ is the rate measure of $\mu^{(\tau)}(\qQ)$ and is given by \cite{Wen-07TCOM}
\begin{equation} \label{eq:Rate_Fun1}
\calR^{(\tau)}(\qQ)=\sum_{q=1}^{L_r}\sum_{p=1}^{L_c}\max_{\tilde{\qQ}_{q,p}}\left\{\frac{1}{\beta}\tr\left(\tilde\qQ_{q,p}\qQ_{q,p}\right)-\frac{1}{\beta N}\log \aangBig{ e^{\tr\left(\tilde\qQ_{q,p}\qX_{p}^H\qT_{q,p}\qX_{p}\right)} }_{\qX_{p}} \right\}
\end{equation}
with $\tilde\qQ_{q,p}\in {\mathbb C}^{(\tau+1)\times(\tau+1)}$ being a symmetric matrix and $\tilde{\qQ}\triangleq \{\tilde\qQ_{q,p}, \forall q, p\}$. Inserting (\ref{eq:Measure_muInN}) into (\ref{eq:sf_E2}) yields $\frac{1}{\beta N}\int e^{\beta N({\cal G}^{(\tau)}(\qQ)-\calR^{(\tau)}(\qQ))+\calO(1)}d\qQ$. Therefore, as $\beta,N\rightarrow\infty$, the integration over $\qQ$ can be performed via the saddle point method, yielding
\begin{equation}\label{eq:QF}
\lim_{\beta, N\rightarrow\infty}\frac{1}{\beta N} \log \aang{\sfZ_{\beta}^{\tau}(\qy,\qA)}_{\qy,\qA} = \max_{\qQ}\left[{\cal G}^{(\tau)}(\qQ)-
\calR^{(\tau)}(\qQ)\right] .
\end{equation}

Substituting (\ref{eq:G3}) and (\ref{eq:Rate_Fun1}) into (\ref{eq:QF}), we arrive the free entropy (\ref{eq:LimF_AppA}) at the saddle-point asymptotic approximation
\begin{equation} \label{eq:RS_FG}
\Phi = \lim_{\tau\rightarrow 0}\frac{\partial}{\partial \tau}\Extr_{\qQ,\tilde{\qQ}}\left\{ \Phi^{(\tau)} \right\},
\end{equation}
where $\Phi^{(\tau)} \triangleq \Phi_1^{(\tau)} + \Phi_2^{(\tau)} + \Phi_3^{(\tau)} + \Phi_4^{(\tau)}$ and
\begin{subequations} \label{eq:TQQ}
\begin{align}
\Phi_1^{(\tau)} &\triangleq \frac{1}{\beta N} \log \aangBigg{ \prod_{q=1}^{L_r}e^{-\tr\left(\left(\sum_{p=1}^{L_c}\qR_{q,p}^{\frac{1}{2}} \qW_{q,p}\qV_{q,p}\right)
{\bf\Sigma}\left(\sum_{p=1}^{L_c}\qR_{q,p}^{\frac{1}{2}} \qW_{q,p}\qV_{q,p}\right)^H\right)} }_{\qW}, \label{eq:TQQ1}\\
\Phi_2^{(\tau)} &\triangleq - \frac{\nu}{\beta} \log\left( 1+\tau\frac{\sigma_0^2}{\sigma^2} \right), \label{eq:TQQ2} \\
\Phi_3^{(\tau)} &\triangleq \frac{1}{\beta N}\sum_{q=1}^{L_r}\sum_{p=1}^{L_c}\log\aangBig{ e^{\tr\left(\tilde\qQ_{q,p}\qX_{p}^H\qT_{q,p}\qX_{p}\right)} }_{\qX_{p}}, \label{eq:TQQ3} \\
\Phi_4^{(\tau)} &\triangleq -\frac{1}{\beta}\sum_{q=1}^{L_r}\sum_{p=1}^{L_c}\tr\left(\tilde\qQ_{q,p}\qQ_{q,p}\right). \label{eq:TQQ4}
\end{align}
\end{subequations}
The saddle-points can be obtained by seeking the points of zero gradient of $\Phi^{(\tau)}$ with respect to $\qQ$ and $\tilde{\qQ}$.

\subsection*{A.1 Replica Symmetry Equations}
Rather than searching for the saddle-points over general forms of $\qQ$ and $\tilde{\qQ}$, we invoke the following hypothesis: The dependence on the replica indices would not affect the physics of the system because replicas have been introduced artificially for the convenience of the expectation operators over $\qy$ and $\qA$. It therefore seems natural to assume {\it replica symmetry} (RS),\footnote{It is natural to set $[\tilde\qQ_{q,p}]_{0,0}=\tilde{r}_{q,p}$ similar to that of $[\qQ_{q,p}]_{0,0}$. It turns out that when $\tau=0$, we get $\tilde{r}_{q,p}=0$. Therefore, to simplify notation, we set $\tilde{r}_{q,p}=0$ at the beginning. In addition, it is natural to let $m_{q,p}$ and $\tm_{q,p}$ be complex-valued variables. We will find that the whole exponents will depend only on the real part of $m_{q,p}$, and $\tm_{q,p}$ turns out to be a real-valued variable. Therefore, we let $m_{q,p}$ and $\tm_{q,p}$ be real-valued variables at the beginning.} i.e.,
\begin{align}
\qQ_{q,p} &=\left[
\begin{array}{cc}
r_{q,p} & m_{q,p} \qone_{\tau}^T \\
m_{q,p} \qone_{\tau} & (Q_{q,p}-q_{q,p})\qI_{\tau} + q_{q,p}\qone_{\tau} \qone_{\tau}^T
\end{array}
\right], \label{eq:RS_Q} \\
\tilde\qQ_{q,p}&= \left[
\begin{array}{cc}
0 & \tm_{q,p} \qone_{\tau}^T \\
\tm_{q,p} \qone_{\tau} & (\tQ_{q,p}-\tq_{q,p})\qI_{\tau} + \tq_{q,p}\qone_{\tau} \qone_{\tau}^T
\end{array}
\right]. \label{eq:RS_tQ}
\end{align}
This RS has been widely accepted in  statistical physics \cite{Nishimori-01BOOK} and used in the field of information/communications theory, e.g., \cite{Kabashima-09JSM,Tanaka-10ISIT,Kabashima-12JSM,Krzakala-12JSM,Vehkapera-arXiv13,Tulino-13IT,Rangan-12TIT,Tanaka-02IT,Moustakas-03TIT,Guo-05IT,Muller-03TSP,Wen-07TCOM,Hatabu-09PRE,Takeuchi-13TIT,Girnyk-14TWC}.

By Lemma \ref{lemma_eigProjectionMatrix}, we can show that for the RS of (\ref{eq:RS_Q}), the eigenvalues of $\qSigma\qQ_{q,p}$ are given by\footnote{The calculation of the eigenvalues can be obtained by using Lemma \ref{lemma_eigProjectionMatrix} for $\qSigma$ and $\qQ_{q,p}$, which is rather laborious but straightforward. For readers' convenience, we detail the calculation in Appendix B.}
\begin{subequations} \label{eq:eigOfSandQ}
\begin{align}
\lambda_0(\qSigma\qQ_{q,p}) &= 0, \\
\lambda_1(\qSigma\qQ_{q,p}) &= \frac{(Q_{q,p}-q_{q,p}) + \tau(r_{q,p}-2m_{q,p} + q_{q,p})}{\sigma^2+\tau\sigma_0^2}, \\
\lambda_{a}(\qSigma\qQ_{q,p}) &= \frac{Q_{q,p}-q_{q,p}}{\sigma^2},~~\mbox{for }a = 2,\ldots,\tau.
\end{align}
\end{subequations}
Write $\qV_{q,p}\qSigma\qV_{q,p}^H = \tqV_{q,p}\tqV_{q,p}^H $, where $\tqV_{q,p} \triangleq \left[\tqv_{q,p}^{(0)} ~\tqv_{q,p}^{(1)} \cdots\tqv_{q,p}^{(\tau)}\right]$ is an $N \times (\tau+1)$ orthogonal matrix. Recall the covariance matrix of $\qV_{q,p}^H \qV_{q,p} $ defined in (\ref{eq:defQ}). From the linear algebra, one can easily obtain that $ \tqv_{q,p}^{(a)}$ is a vector with length $N \lambda_{a}(\qSigma\qQ_{q,p})$ for $a = 0,1,\ldots, \tau$. Applying Lemma \ref{lemma_HaarIntegr} to (\ref{eq:TQQ1}), we get
\begin{equation}\label{eq:RS_G}
\Phi_1^{(\tau)} =\sum_{q=1}^{L_r} H_{q}\left( \left\{\frac{(Q_{q,p}-q_{q,p}) + \tau(r_{q,p}-2m_{q,p} + q_{q,p})}{\sigma^2+\tau\sigma_0^2} \right\} \right)+ (\tau-1) H_{q}\left( \left\{\frac{Q_{q,p}-q_{q,p}}{\sigma^2} \right\} \right),
\end{equation}
where
\begin{equation} \label{eq:defGq}
H_{q}(\{ x_{q,p} \}) \triangleq \Extr_{\{\Gamma_{q,p}\}} \left\{ \Bigg(\sum_{p=1}^{L_c} \Gamma_{q,p} x_{q,p}- \log\Gamma_{q,p}x_{q,p}-1 \Bigg) - \nu_{q}\log\Bigg(1 + \sum_{p=1}^{L_c} \frac{R_{q,p}}{\Gamma_{q,p}} \Bigg) \right\} + \calO\left(\frac{1}{N}\right).
\end{equation}
The solution to the extremization problem in (\ref{eq:defGq}), denoted by $\{ \Gamma_{q,p}^{\star} \}$, enforces the condition
\begin{equation} \label{eq:extrGq}
\Gamma_{q,p}^{\star} - \frac{1}{x_{q,p}} = -\frac{\Delta_{q,p}}{x_{q,p}},
\end{equation}
where
\begin{equation} \label{eq:defDelta}
\Delta_{q,p} \triangleq \nu_{q} \frac{\frac{R_{q,p}}{\Gamma_{q,p}^{\star}}}{1 + \sum_{l=1}^{L_c} \frac{R_{q,l}}{\Gamma_{q,l}^{\star}}}.
\end{equation}

Next, we focus the RS calculation of (\ref{eq:TQQ3}). Substituting the RS form for $\tqQ_{q,p}$ in (\ref{eq:RS_Q}) and the definition of $\qT_{q,p}$ in (\ref{eq:def_RT}) into (\ref{eq:TQQ3}), we obtain
\begin{align}
\Phi_3^{(\tau)}
=&\frac{1}{\beta N}\log \aangBig{ \prod_{p=1}^{L_c} e^{\tr\left({\tt vec}\left(\qX_{p}\right)^H\left(\sum_{q=1}^{L_r} (\tilde\qQ_{q,p}\otimes\qT_{q,p}) \right)
{\tt vec}\left(\qX_{p}\right)\right)} }_{\qX} \notag \\
=& \frac{1}{\beta N}\log \aangBig{ \prod_{p=1}^{L_c} e^{ N \mu_p \left( |\sum_{a=1}^\tau\sqrt{\tq_{p}}x_{p}^{(a)}|^2
  + \sum_{a=1}^\tau  \left( 2 \tm_{p} {\rm Re}\left\{(x_{p}^{(a)})^* x_{p}^{(0)} \right\} + (\tQ_{p}-\tq_{p})|x_{p}^{(a)}|^2 \right) \right) } }_{\qX}, \label{eq:LT_Rate_Fun1}
\end{align}
where
\begin{equation}
\left\{\begin{aligned}
\tQ_{p} &\triangleq \sum_{q=1}^{L_r}\tQ_{q,p},\\
\tq_{p} &\triangleq \sum_{q=1}^{L_r}\tq_{q,p},\\
\tm_{p} &\triangleq \sum_{q=1}^{L_r} \tm_{q,p}.
\end{aligned}\right.
\end{equation}
Then we decouple the first quadratic term in the exponent of (\ref{eq:LT_Rate_Fun1}) by using the Hubbard-Stratonovich transformation (Lemma \ref{lemma_GaussianIntegr}) and introducing the auxiliary vector $z_{p}$, to rewrite (\ref{eq:LT_Rate_Fun1}) as
\begin{multline}\label{eq:LT_Rate_Fun2}
\frac{1}{\beta N}  \log \aangBig{ \int \prod_{p=1}^{L_c} \rmD z_{p}\,e^{N \mu_p \left( \left(\sum_{a=1}^\tau\sqrt{\tq_{p}}x_{p}^{(a)}\right)^* z_{p} + z_{p}^*\left(\sum_{a=1}^\tau\sqrt{\tq_{p}} x_{p}^{(a)}\right) + \sum_{a=1}^\tau \left( 2 \tm_{p} {\rm Re}\left\{(x_{p}^{(a)})^* x_{p}^{(0)} \right\} - (\tq_{p}-\tQ_{p})|x_{p}^{(a)}|^2 \right) \right) } }_{\qX}\\
= \frac{1}{\beta N}  \log \prod_{p=1}^{L_c} \aangBig{ \int \rmD z_{p}  \left( \int \rmd x_{p} \,e^{-N \mu_p
   \left( (\tq_{p}-\tQ_{p})|x_{p}|^2 - 2\dRe\{ x_{p}^* (\tm_{p}x_{p}^{(0)} + \sqrt{\tq_{p}}z_{p}) \} + \beta |x_{p}| \right) } \right)^{\tau} }_{x_{p}^{(0)}},
\end{multline}
where the equality follows from the fact that $\qx^{(a)}$ is a random vector taken from the input distribution $P_\beta(\qx) = e^{- \beta \|\qx\|_1 }$ if $a \neq 0$. Lastly, using (\ref{eq:RS_Q}) and (\ref{eq:RS_tQ}) into (\ref{eq:TQQ4}), we obtain
\begin{equation} \label{eq:RSQQ}
\Phi_4^{(\tau)}= -\frac{1}{\beta}\sum_{q=1}^{L_r}\sum_{p=1}^{L_c} \Big( 2\tau m_{q,p}\tm_{q,p} -\tau (Q_{q,p}-q_{q,p})\tq_{q,p} - \tau Q_{q,p}(\tq_{q,p}-\tQ_{q,p}) + \tau^2q_{q,p}\tq_{q,p} \Big).
\end{equation}

Recall that we have denoted $\sigma^2 = \lambda/\beta$. Before proceeding, we introduce the rescaled variables as
\begin{equation}
\left\{\begin{aligned}
\chi_{q,p} &= \beta(Q_{q,p}-q_{q,p}),\\
\hQ_{q,p} &= (\tq_{q,p}-\tQ_{q,p})/\beta,\\
\hchi_{q,p} &= \tq_{q,p}/\beta^2,\\
\hm_{q,p} &= \tm_{q,p}/\beta,
\end{aligned}\right.
\end{equation}
and we define $\hQ_{p} \triangleq \sum_{q=1}^{L_r} \hQ_{q,p}$, $\hchi_{p} \triangleq \sum_{q=1}^{L_r} \hchi_{q,p}$, $\hm_{p} \triangleq \sum_{q=1}^{L_r} \hm_{q,p}$. Using these variables into (\ref{eq:RS_G}), (\ref{eq:LT_Rate_Fun2}), (\ref{eq:RSQQ}), we obtain
\begin{align}
\Phi_1^{(\tau)} &=\sum_{q=1}^{L_r} H_{q}\left( \left\{\frac{\chi_{q,p} + \tau \beta (r_{q,p}-2m_{q,p} + q_{q,p})}{\lambda+\tau \beta \sigma_0^2} \right\} \right) + (\tau-1) H_{q}\left( \left\{\frac{\chi_{q,p}}{\lambda} \right\} \right), \label{eq:RS_Phi1_reScaled} \\
\Phi_3^{(\tau)} &=\frac{1}{\beta N}  \log \int \prod_{p=1}^{L_c} \rmD z_{p}\, \aangBig{ \left( \phi(x_{p};\hQ_{p},\hm_{p}x_{p}^{(0)} + \sqrt{\hchi_{p}}z_{p},\beta N \mu_p) \right)^{\tau} }_{x_{p}^{(0)}}, \label{eq:RS_Phi3_reScaled} \\
\Phi_4^{(\tau)}&= -\sum_{q=1}^{L_r}\sum_{p=1}^{L_c} \Big( 2\tau m_{q,p}\tm_{q,p} + \tau \chi_{q,p} \hchi_{q,p} - \tau Q_{q,p}\hQ_{q,p} + \tau^2 \beta \hchi_{q,p} q_{q,p} \Big), \label{eq:RS_Phi4_reScaled}
\end{align}
where we have defined
\begin{equation} \label{eq:defPhi}
\phi(x;a,b,c) \triangleq \int \rmd x \, e^{ -c \left( a |x|^2 - 2{\rm Re}\{ b^*x \} + |x| \right) } .
\end{equation}

Substituting (\ref{eq:RS_Phi1_reScaled})--(\ref{eq:RS_Phi4_reScaled}) into (\ref{eq:RS_FG}), taking the derivative of $\Phi^{(\tau)}$ with respect to $\tau$, and letting $\tau \rightarrow 0$, we find $\Phi = \Phi_1 + \Phi_2 + \Phi_3 + \Phi_4$, where
\begin{subequations} \label{eq:GenFree_Beta}
\begin{align}
\Phi_1 &\triangleq \sum_{q=1}^{L_r} \sum_{p=1}^{L_c} \left( \frac{(r_{q,p}-2m_{q,p} + Q_{q,p})}{\lambda}
- \frac{\sigma_0^2 \chi_{q,p}}{\lambda^2} \right) \left(\Gamma_{q,p}^{\star} - \frac{\lambda}{\chi_{q,p}} \right) + \sum_{q=1}^{L_r} \frac{1}{\beta} H_{q}\left( \left\{ \frac{\chi_{q,p}}{\lambda} \right\} \right),
 \label{eq:GenFree_Beta2} \\
\Phi_2 &\triangleq - \frac{\nu \sigma_0^2}{\lambda}, \\
\Phi_3 &\triangleq \frac{1}{\beta N} \int \prod_{p=1}^{L_c} \rmD z_{p} \, \aangBig{
 \log \phi\left(x_{p};\hQ_{p},\hm_{p}x_{p}^{(0)} + \sqrt{\hchi_{p}}z_{p},\beta N \mu_p\right) }_{x_{p}^{(0)}}, \label{eq:GenFree_Beta4} \\
\Phi_4 &\triangleq \sum_{p=1}^{L_c}\sum_{q=1}^{L_r} (Q_{q,p}\hQ_{q,p} - \chi_{q,p}\hchi_{q,p} - 2m_{q,p}\hm_{q,p}).
\end{align}
\end{subequations}
In (\ref{eq:GenFree_Beta2}), we have used the fact that $\frac{q_{q,p}}{\beta(Q_{q,p}-q_{q,p})} = \frac{Q_{q,p} - (Q_{q,p}-q_{q,p}) }{\beta(Q_{q,p}-q_{q,p})}\rightarrow \frac{Q_{q,p}}{\chi_{q,p}}$ as $\beta \rightarrow \infty$. Also, we have used the following result
\begin{equation}\label{eq:difGq}
\left.\frac{\partial H_{q}(\{x_{q,k}\})}{\partial x_{q,p}}\right|_{x_{q,p} = \frac{\chi_{q,p}}{\lambda}}= \Gamma_{q,p}^{\star} - \frac{\lambda}{\chi_{q,p}},
\end{equation}
where the equality follows directly from taking the derivative of $H_{q}$ in (\ref{eq:defGq}). Notice that when substituting $x_{q,p} =\frac{\chi_{q,p}}{\lambda}$ into (\ref{eq:extrGq}), we get
\begin{equation}
 \Gamma_{q,p}^{\star} - \frac{\lambda}{\chi_{q,p}} = - \frac{\Delta_{q,p}\lambda}{\chi_{q,p}}.
\end{equation}
This identity will be used later to simplify some expressions. Also, as $\beta \rightarrow \infty$, $\frac{1}{\beta} H_{q}\left( \left\{\frac{\chi_{q,p}}{\lambda} \right\} \right) \rightarrow 0$ and $\Phi_2 \rightarrow 0$.

Now, recall that we have to search $\qQ$ and $\tilde{\qQ}$ which achieve the extremal condition in (\ref{eq:RS_FG}). With the RS assumption, we only have to determine $\{\hchi_{q,p},\hQ_{q,p},\hm_{q,p},\}$, which can be obtained by equating the partial derivatives of $\Phi^{(\tau)}$ to zeros, i.e.,
\begin{equation}\label{eq:sadp}
\frac{\partial\Phi^{(\tau)}}{\partial \chi_{q,p}}
=\frac{\partial\Phi^{(\tau)}}{\partial Q_{q,p}}
=\frac{\partial\Phi^{(\tau)}}{\partial m_{q,p}}, ~~\forall q, p,
\end{equation}
and then letting $\tau \rightarrow 0$. Evaluating these calculations, we obtain
\begin{subequations} \label{eq:sdPoint_Beta}
\begin{align}
\hQ_{q,p} &= \frac{\Delta_{q,p}}{\chi_{q,p}} , \label{eq:sdPoint_Beta-1}\\
\hchi_{q,p} &= \sum_{r=1}^{L_c} \left(\frac{\mse_{q,r}}{\lambda} - \frac{\sigma_0^2\chi_{q,r}}{\lambda^2} \right)\Gamma'_{q,p,r}
+ \frac{\mse_{q,p}}{\chi_{q,p}^{2}} - \frac{\sigma_0^2}{\lambda} \frac{(1-\Delta_{q,p})}{\chi_{q,p}}, \label{eq:sdPoint_Beta-2} \\
\hQ_{q,p} &= \hm_{q,p}, \label{eq:sdPoint_Beta-3}
\end{align}
\end{subequations}
where $\Gamma'_{q,p,r} \triangleq \partial \Gamma_{q,p}^{\star}/\partial \chi_{q,r}$ and
\begin{equation} \label{eq:defMMSE_qp_Beta}
\mse_{q,p} \triangleq  r_{q,p}-2m_{q,p} + Q_{q,p}.
\end{equation}

Following \cite{Vehkapera-arXiv13}, the expression of $\Gamma'_{q,p,r}$ can be obtained via the inverse function theory
\begin{equation}
\qGamma_{q}' = \left[ \frac{\partial (\chi_{q,1},\ldots, \chi_{q,L_c})}{\partial (\Gamma_{q,1}^{\star},\ldots,\Gamma_{q,L_c}^{\star}) }\right]^{-1},
\end{equation}
where $\frac{\partial (\chi_{q,1},\ldots, \chi_{q,L_c})}{\partial (\Gamma_{q,1}^{\star},\ldots,\Gamma_{q,L_c}^{\star}) }$ is the Jacobian matrix with its $(i,j)$th element being
\begin{align}
\frac{\partial \chi_{q,i}}{\partial \Gamma_{q,j}^{\star} }
&= \frac{\partial }{\partial \Gamma_{q,j}^{\star} } \frac{\lambda(1-\Delta_{q,i})}{\Gamma_{q,i}^{\star}} \notag \\
&= -\lambda \left(\frac{(1-\Delta_{q,i})}{\Gamma_{q,i}^{\star 2}} \delta_{i,j} + \frac{1}{\Gamma_{q,i}^{\star}} \frac{\partial \Delta_{q,i} }{\partial \Gamma_{q,j}} \right) \notag \\
&= -\lambda \left(\frac{(1-2\Delta_{q,i})}{\Gamma_{q,i}^{\star 2}} \delta_{i,j} + \frac{1}{\nu_{q}} \frac{\Delta_{q,i}}{\Gamma_{q,i}^{\star}} \frac{\Delta_{q,j}}{\Gamma_{q,j}^{\star}} \right),
\end{align}
where the first equality follows from (\ref{eq:difGq}). In addition, $\qGamma_{q}'$ can be explicitly obtained by applying the matrix inverse lemma. Specifically, we have
\begin{equation}
\Gamma'_{q,p,r} = [\qGamma_{q}']_{p,r} = \frac{1}{\lambda} \left(\frac{1}{\nu_{q}} \frac{\Delta_{q,p}\Delta_{q,r}\Gamma_{q,p}\Gamma_{q,r}}{(1-2\Delta_{q,p})(1-2\Delta_{q,r})} \left( 1+ \sum_{l=1}^{L_c} \frac{1}{\nu_{q}} \frac{\Delta_{q,l}^2}{1-2\Delta_{q,l}}\right)^{-1}
- \frac{\Gamma_{q,r}^2}{1-2\Delta_{q,r}} \delta_{p,r} \right). \label{eq:GammaMatrixInvLemma}
\end{equation}

To get more explicit expressions for $\mse_{q,p}$ and $\chi_{q,p}$, let us simplify $\Phi_3$ in (\ref{eq:GenFree_Beta4}). As $\beta \rightarrow \infty$, we obtain
\begin{equation} \label{eq:phiBeta_minXp}
\frac{1}{\beta N} \log \phi(x_{p};\hQ_{p},\hm_{p}x_{p}^{(0)} + \sqrt{\hchi_{p}}z_{p},\beta N \mu_p)
= \mu_p \min_{ x_{p} } \Big\{ \hm_p |x_{p}|^2 - x_{p}^* y_p - y_p^*x_{p} + |x_{p}|  \Big\},
\end{equation}
where we have used $y_p = \hm_{p}x_{p}^{(0)} + \sqrt{\hchi_{p}}z_{p}$ and have used the identity $\hQ_{q,p} = \hm_{q,p}$ in (\ref{eq:sdPoint_Beta-3}) to simplify the result. The optimal solution $\hat{x}_{p}$ in (\ref{eq:phiBeta_minXp}) is given by (see \cite[Lemma V.1]{Maleki-13IT} for a derivation)
\begin{equation} \label{eq:optOfX}
 \hat{x}_p = \frac{\left( |y_p| - \frac{1}{2}\right)_{+} \frac{y_p}{|y_p|}}{\hm_{p}}.
\end{equation}
If we substitute the optimal solution (\ref{eq:optOfX}) into (\ref{eq:phiBeta_minXp}), we get
\begin{equation} \label{eq:phiBeta_minXp2}
 (\ref{eq:phiBeta_minXp}) = - \mu_p  \frac{\left(|y_p|-\frac{1}{2}\right)^2}{\hm_{p}}\, {\sf I}_{ \{|y_p| > \frac{1}{2}\} }
= - \mu_p G(y_p; \hm_{p}),
\end{equation}
where we have defined
\begin{equation}
G(y; A) =  \frac{\left(|y|-\frac{1}{2}\right)^2}{A} \,{\sf I}_{ \{|y| > \frac{1}{2}\} }.
\end{equation}

Notice that $\qx_{p}^{(0)}$ and $\qz_p$ are standard Gaussian random vectors with i.i.d.~entries for all $p$. Therefore, let $Z$ denote the standard Gaussian random random variable. Then (\ref{eq:GenFree_Beta4}) turns out to be
\begin{equation} \label{eq:GenFree_Beta4-1}
\Phi_3 = - \sum_{p=1}^{L_c} \mu_p \aangBig{ G\left(Z\sqrt{\hchi_{p}+ \hm_{p}^2} ; \hm_{p}\right) }_{Z} .
\end{equation}

To deal with the partial derivatives of $\Phi_3$, let us first define
\begin{align}
g_{\sfc}(\zeta) &\triangleq \zeta e^{-\frac{1}{4\zeta}} - \sqrt{\pi\zeta} \sfQ\left(\frac{1}{\sqrt{2\zeta}}\right), \\
\aag_{\sfc}(\zeta) &\triangleq e^{-\frac{1}{4\zeta}} - \sqrt{\frac{\pi}{4\zeta}} \sfQ\left(\frac{1}{\sqrt{2\zeta}}\right).
\end{align}
Following the manipulations as those in \cite[(350)--(353)]{Tulino-13IT}, we have the following useful identities:
\begin{align}
 \aang{ G(Z \sqrt{\zeta} ; A) }_{Z}
 &= \frac{1}{A} \int_{|z| > 1/2} \left(|z|-\frac{1}{2}\right)^2 \frac{1}{\pi\zeta} e^{-\frac{1}{\zeta} |z|^2} \rmd z = \frac{g_{\sfc}(\zeta)}{A}, \label{eq:ExG}
\end{align}
and
\begin{align}
\frac{\partial g_{\sfc}(\zeta) }{ \partial x}
&= \left( \frac{\partial \zeta}{\partial x} \right) \aag_{\sfc}(\zeta) .
\end{align}
After assessing the partial derivatives of $\Phi$ (or more precisely $\Phi_3+\Phi_4$ ) with respect to the variables $\{\hm_{q,p},\hQ_{q,p},\hchi_{q,p}\}$, we obtain
\begin{subequations} \label{eq:sdPoint_Beta2}
\begin{align}
m_{q,p} &= \mu_p \, \rho_x \aag_{\sfc}(\hm_{p}^2+\hchi_{p}) , \\
Q_{q,p} &= \mu_p \left( \frac{1-\rho_x}{\hm_{p}^2} g_{\sfc}(\hchi_{p}) + \frac{\rho_x}{\hm_{p}^2} g_{\sfc}(\hm_{p}^2+\hchi_{p}) \right), \\
\chi_{q,p} &= \mu_p \left( \frac{1-\rho_x}{\hm_{p}} \aag_{\sfc}(\hchi_{p}) + \frac{\rho_x}{\hm_{p}} \aag_{\sfc}(\hm_{p}^2+\hchi_{p})  \right).
\end{align}
\end{subequations}
In addition, directly from the definition of $\qQ_{q,p}$ in (\ref{eq:defQ}), we have $r_{q,p} = \mu_p \rho_x$. Substituting $r_{q,p}$ and (\ref{eq:sdPoint_Beta}) into (\ref{eq:defMMSE_qp_Beta}), we obtain $\mse_{q,p}$. Notice that $m_{q,p}$, $Q_{q,p}$, $\chi_{q,p}$ and $\mse_{q,p}$ are irrelevant to index $q$. We denote $m_p = m_{q,p}$, $Q_p = Q_{q,p}$, $\chi_p = \chi_{q,p}$, and $\mse_p = \mse_{q,p}$ for clarity.

Combining the definition in (\ref{eq:defDelta}), the result in (\ref{eq:GammaMatrixInvLemma}), and all the coupled equations in (\ref{eq:sdPoint_Beta}) and (\ref{eq:sdPoint_Beta2}), we get the result in Proposition \ref{Pro1}. Notice that in Proposition \ref{Pro1}, we have used the rescaled variable $\Gamma_{q,p}^{\star}/\lambda \rightarrow \Gamma_{q,p}^{\star}$ for the sake of notational simplicity.

\section*{\sc Appendix B: Eigenvalues of the matrix $\qSigma\qQ_{q,p}$}
Applying Lemma \ref{lemma_eigProjectionMatrix} to $\qSigma$ and $\qQ_{q,p}$, we get
\begin{equation}
\qSigma\qQ_{q,p} = \qU \left[
\begin{array}{cc}
\qB_{1,1} & \qzero \\
\qzero & \qB_{2,2} \\
\end{array}
\right] \qU^H,
\end{equation}
where
\begin{subequations}
\begin{align}
\qB_{1,1} &= \frac{1}{\sigma^2+\tau\sigma_0^2}
\left[\begin{array}{cc}
\tau(r_{q,p}-m_{q,p}^{*}) & \sqrt{\tau}\big(\tau(m_{q,p}-q_{q,p})-(Q_{q,p}-q_{q,p})\big) \\
-\sqrt{\tau}(r_{q,p}-m_{q,p}^{*}) & -\tau(m_{q,p}-q_{q,p})-(Q_{q,p}-q_{q,p})
\end{array}
\right], \\
\qB_{2,2} &= \frac{Q_{q,p}-q_{q,p}}{\sigma^2} \qI_{\tau-1}.
\end{align}
\end{subequations}
It is easy to see that the rows of $\qB_{1,1}$ are linearly dependent, and the eigenvalues are $\frac{r_{q,p}-m_{q,p}-m_{q,p}^{*} + q_{q,p}}{\sigma^2+\tau\sigma_0^2}$ and $0$. Therefore, the eigenvalues of $\qSigma\qQ_{q,p}$ are in the form of (\ref{eq:eigOfSandQ}).

\section*{\sc Appendix C: Mathematical Tools}
For readers' convenience, we provide some mathematical tools needed in this appendix.

\begin{Lemma} (Gaussian Integral and Hubbard-Stratonovich Transformation) \label{lemma_GaussianIntegr}
Let $\qz$ and $\qb$ be $N$-dimensional real vectors, and $\qA$ an $M \times M$ positive definite matrix. Then
\begin{equation}
\frac{1}{\pi^N} \int \rmd \qz e^{-\qz^H\qA\qz+\qz^H\qb+\qb^H\qz}
= \frac{1}{\det(\qA)} e^{\qb^H\qA^{-1}\qb}.
\end{equation}
Using this equation from right to left is usually called the \emph{Hubbard-Stratonovich} transformation.

%It allows an expression with $\qb$ in the exponent to be replaced by one with $\qb$ occurring linearly in the exponent at the expense of an additional integration over an auxiliary variable $\qz$. This is particularly useful to factorize integrals when $\qb$ is a sum of integration variables.
\end{Lemma}

\begin{Lemma} \label{lemma_eigProjectionMatrix}
For a matrix
\begin{equation}
\qA= \left[
\begin{array}{cc}
A_{1,1} & A_{1,2} \qone_{\tau}^T \\
A_{1,2}^* \qone_{\tau} & A_{2,2}\qI_{\tau} + \epsilon \qone_{\tau} \qone_{\tau}^T
\end{array}
\right] \in \bbC^{(\tau+1) \times (\tau+1)},
\end{equation}
the eigen-decomposition of the matrix is given by \cite{Kabashima-08JPCS}
\begin{equation}
\qA= \qU \left[
\begin{array}{cc|cccc}
A_{1,1} & \sqrt{\tau}A_{1,2} & 0 & 0 & \ldots & 0 \\
\sqrt{\tau}A_{1,2}^{*} & A_{2,2} + \tau \epsilon & 0 & 0 & \ldots & 0 \\
\hline
0 & 0 & A_{2,2} & 0 & \ldots & 0 \\
0 & 0 & 0 & A_{2,2} & \ldots & 0 \\
\vdots & \vdots & \vdots & \vdots & \ddots & \vdots \\
0 & 0 & 0 & 0 & \ldots & A_{2,2}
\end{array}
\right] \qU^H,
\end{equation}
where $\qU = [\qu_0 ~\qu_1\cdots \qu_{\tau} ]$ denotes a $(\tau+1)$-dimensional orthonormal basis composed of $\qu_0 =[1~0~0\cdots0]^T$, $\qu_1 =[0~\tau^{-1/2}~\tau^{-1/2}\cdots\tau^{-1/2}]^T$ and $\tau-1$ orthonormal vectors $\qu_2,\qu_3,\ldots,\qu_{\tau}$, which are orthogonal to both $\qu_0$ and $\qu_1$.
\end{Lemma}

\begin{Lemma} \label{lemma_HaarIntegr}
Let $\{ \qx_{p}: p=1,\ldots,L \}$ be a set of vectors that satisfy $\| \qx_{p} \|^2 = N x_{p} $ for some non-negative real values $\{ x_{p} \}$, $\{\qW_{p} \in\bbC^{N \times N}: p= 1,\ldots,L\}$ be a set of independent Haar measure of random matrices, and $\{\qR_{p}\}$ be a set of positive-semidefinite matrices. Define
\begin{equation} \label{eq:integrOfOrthRM}
H(\{ x_{p} \}) = \frac{1}{N} \log \aangBigg{ e^{ -\frac{1}{\sigma^2} \left\| \sum_{p=1}^{L}\qR_{p}^{\frac{1}{2}}  \qW_{p} \qx_{p} \right\|^2 } }_{\{\qW_{p} \}} .
\end{equation}
Then for large $N$, we have
\begin{equation}
H(\{ x_{p} \}) =
\Extr_{\{\Gamma_{p}\}} \left\{ \sum_{p=1}^{L} (\Gamma_{p} x_{p}- \log\Gamma_{p} x_{p}-1)-\frac{1}{N} \log\det\left(\qI_N + \sum_{p=1}^{L} \frac{1}{\sigma^2\Gamma_{p}}\qR_{p} \right)  \right\}
+ \calO(1/N) .
\end{equation}
This lemma extends \cite[Lemma 1]{Vehkapera-arXiv13} to deal with the formula of (\ref{eq:integrOfOrthRM}) when $\qR_{p} \neq \qI_N $ and $\{\qW_{p}\}$ are the Haar measure of \emph{complex} random matrices.
\end{Lemma}

\begin{proof}
This lemma can be obtained by following the steps of \cite[Lemma 1]{Vehkapera-arXiv13} with some variations in order to adopt it to our setting. 

From the definition of $\qW_{p}$ and $\qx_{p}$, the vector $\qu_{p} = \qW_{p} \qx_{p}$ can be considered to be uniformly distributed on a surface of a sphere with
radius $\sqrt{N x_{p}}$ for each $p$. Then the joint probability density function (pdf) of $\{ \qu_{p} \}$ is given by
\begin{align}
P(\{ \qu_{p} \})
&= \frac{1}{\sfZ} \prod_{p=1}^{L} \delta\left( \|\qu_{p}\|^2 - N x_{p} \right) \notag \\
&= \frac{1}{\sfZ} \int  \prod_{p=1}^{L}\left( \rmd\Gamma_{p} \frac{1}{2 \pi \sfj} e^{-\Gamma_{p} (\| \qu_{p}\|^2 - N x_{p}) } \right), \label{eq:pu_{p}}
\end{align}
where $\sfZ$ is the normalization factor and $\{ \Gamma_{p} \}$ is a set of complex numbers. The normalization factor is given by
\begin{equation}
\sfZ = \int  \prod_{p=1}^{L}\left( \rmd\Gamma_{p} \rmd\qu_{p} \frac{1}{2 \pi \sfj} e^{-\Gamma_{p} (\| \qu_{p}\|^2 - N x_{p}) } \right).
\end{equation}
Using the Gaussian integration formula (i.e., Lemma \ref{lemma_GaussianIntegr}) with respect to $\qu_{p}$, the normalization factor becomes
\begin{equation}
\sfZ = \int  \prod_{p=1}^{L}\left( \rmd\Gamma_{p} \frac{\pi^{N}}{2 \pi \sfj} e^{ N (\Gamma_{p} x_{p} - \log\Gamma_{p}) } \right).
\end{equation}
Since we are interested in the large $N$ analysis, the saddle-point method can further simplify the normalization factor to the form
\begin{align}
\frac{1}{N}\log\sfZ &= \sum_{p=1}^{L} \Extr_{\Gamma_{p}} \left\{ \Gamma_{p} x_{p} - \log\Gamma_{p} \right\} + \log\pi + \calO(1/N) \notag \\
&= \sum_{p=1}^{L} (1+\log x_{p} ) + \log\pi + \calO(1/N), \label{eq:Zfinal}
\end{align}
where the second equality is obtained by solving the extremization problem.

Next, we deal with the calculation of $H$ by writting
\begin{align}
H &=  \frac{1}{N}\log \aangBigg{ e^{ -\frac{1}{\sigma^2} \left\| \sum_{p=1}^{L}\qR_{q,p}^{\frac{1}{2}}  \qW_{p} \qx_{p} \right\|^2 } }_{\{ \qW_{p} \}} \notag \\
  &=  \frac{1}{N}\log \left( \int \prod_{p=1}^{L} \rmd\qu_{p} P(\{ \qu_{p} \})
      e^{ -\frac{1}{\sigma^2} \left\| \sum_{p=1}^{L}\qR_{q,p}^{\frac{1}{2}}  \qu_{p} \right\|^2 } \right),
\end{align}
where the second equality follows from the definition of the joint pdf of $\{ \qu_{p} \}$. Applying the Hubbard-Stratonovich transformation (Lemma \ref{lemma_GaussianIntegr}) together with the expressions (\ref{eq:pu_{p}}) and (\ref{eq:Zfinal}) to the above provides
\begin{multline}
H= \frac{1}{N}\log \left[ \int \left(\prod_{p=1}^{L} \rmd\Gamma_{p} e^{N \Gamma_{p} x_{p}}\right) \int \rmd\qz e^{-\sigma^2 \qz^H\qz}
\times ~\int \prod_{p=1}^{L} \rmd\qu_{p} ~e^{-\Gamma_{p} \qu_{p}^H\qu_{p} + \sfj \qz^H(\qR_{q,p}^{\frac{1}{2}}  \qu_{p}) - \sfj (\qR_{q,p}^{\frac{1}{2}}  \qu_{p})^H\qz } \right]\\
+ \log\frac{\sigma^2}{\pi}  - \frac{1}{N}\log \sfZ.
\end{multline}
Using the Gaussian integration repeatedly with respect to $\{ \qu_{p} \}$ and $\qz$ yields
\begin{align}
H&= \frac{1}{N} \log \int \left(\prod_{p=1}^{L} \rmd\Gamma_{p} e^{N (\Gamma_{p} x_{p}-\log\Gamma_{p}) }\right)
\int \rmd\qz ~e^{- \qz^H\left(\sigma^2\qI_N + \sum_{p=1}^{L} \frac{1}{\Gamma_{p}}\qR_{p} \right)\qz } \notag \\
& \qquad + \log\frac{\sigma^2}{\pi} -\sum_{p=1}^{L} (1+\log x_{p} )
\notag \\
&= \frac{1}{N}\log \int \left(\prod_{p=1}^{L} \rmd\Gamma_{p}\right) e^{ \sum_{p=1}^{L}N (\Gamma_{p} x_{p}- \log\Gamma_{p})-\log\det\left(\sigma^2\qI_N + \sum_{p=1}^{L} \frac{1}{\Gamma_{p}}\qR_{p} \right) } \notag \\
& \qquad + \log\sigma^2 -\sum_{p=1}^{L} (1+\log x_{p} )
\notag \\
&= \Extr_{\{\Gamma_{p}\}} \left\{ \sum_{p=1}^{L} (\Gamma_{p} x_{p}- \log\Gamma_{p})-\frac{1}{N}\log\det\left(\sigma^2\qI_N + \sum_{p=1}^{L} \frac{1}{\Gamma_{p}}\qR_{p} \right)  \right\} \notag \\
& \qquad + \log\sigma^2 -\sum_{p=1}^{L} (1+\log x_{p} ) + \calO(1/N)
\notag \\
&= \Extr_{\{\Gamma_{p}\}} \left\{ \sum_{p=1}^{L} (\Gamma_{p} x_{p}- \log\Gamma_{p} x_{p}-1)-\frac{1}{N}\log\det\left(\qI_N + \sum_{p=1}^{L} \frac{1}{\sigma^2\Gamma_{p}}\qR_{p} \right)  \right\}  + \calO(1/N),
\end{align}
where the third equality is obtained by applying the saddle-point method.
\end{proof}

{\renewcommand{\baselinestretch}{1.1}

}


\begin{thebibliography}{10}
\providecommand{\url}[1]{#1}
\csname url@samestyle\endcsname
\providecommand{\newblock}{\relax}
\providecommand{\bibinfo}[2]{#2}
\providecommand{\BIBentrySTDinterwordspacing}{\spaceskip=0pt\relax}
\providecommand{\BIBentryALTinterwordstretchfactor}{4}
\providecommand{\BIBentryALTinterwordspacing}{\spaceskip=\fontdimen2\font plus
\BIBentryALTinterwordstretchfactor\fontdimen3\font minus
  \fontdimen4\font\relax}
\providecommand{\BIBforeignlanguage}[2]{{%
\expandafter\ifx\csname l@#1\endcsname\relax
\typeout{** WARNING: IEEEtran.bst: No hyphenation pattern has been}%
\typeout{** loaded for the language `#1'. Using the pattern for}%
\typeout{** the default language instead.}%
\else
\language=\csname l@#1\endcsname
\fi
#2}}
\providecommand{\BIBdecl}{\relax}
\BIBdecl

\bibitem{Candes-05IT}
{E. Cand\`{e}s and T. Tao}, ``{Decoding by linear programming},'' \emph{IEEE
  Trans. Inf. Theory}, vol.~51, no.~12, pp. 4203--4215, Dec. 2005.

\bibitem{Donoho-06IT}
{D. L. Donoho}, ``{Compressed sensing},'' \emph{IEEE Trans. Inf. Theory},
  vol.~52, no.~4, pp. 1289--1306, Apr. 2006.

\bibitem{Tropp-10ProcIEEE}
{J. A. Tropp and S. J. Wright}, ``{Computational methods for sparse solution of
  linear inverse problems},'' \emph{Proc. IEEE}, vol.~98, no.~6, pp. 948--958,
  Jun. 2010.

\bibitem{Hayashi-13IEICE}
{K. Hayashi, M. Nagahara, and T. Tanaka}, ``{A user's guide to compressed
  sensing for communications systems},'' \emph{IEICE Trans. Commun.}, vol.
  E96-B, no.~3, pp. 685--712, Mar. 2013.

\bibitem{Tibshirani-96JRSS}
{R. Tibshirani}, ``{Regression shrinkage and selection via the lasso},''
  \emph{J. Royal. Statist. Soc., Ser. B}, vol.~58, no.~1, pp. 267--288, 1996.

\bibitem{Maleki-13IT}
{A. Maleki, L. Anitori, Z. Yang, and R. G. Baraniuk}, ``{Asymptotic analysis of
  complex LASSO via complex approximate message passing (CAMP)},'' \emph{IEEE
  Trans. Inf. Theory}, vol.~59, no.~7, pp. 4290--4308, Jul. 2013.

\bibitem{Parikh-14FTinOpt}
N.~Parikh and S.~Boyd, \emph{{Proximal Algorithms}}.\hskip 1em plus 0.5em minus
  0.4em\relax Foundations and Trends in Optimization, 2014.

\bibitem{Candes-08}
{E. J. Cand\'{e}s}, ``{The restricted isometry property and its implications
  for compressed sensing},'' \emph{C. R. l' Academie des Sciences, ser. I},
  vol. 346, no. 346, pp. 589--592, 2008.

\bibitem{Donoho-05PNAS}
{D. L. Donoho and J. Tanner}, ``{Sparse nonnegative solution of underdetermined
  linear equations by linear programming},'' \emph{Proc. Nat. Acad. Sci.},
  2005.

\bibitem{Donoho-10ProcIEEE}
{D. L. Donoho and Jared Tanner}, ``{Precise undersampling theorems},''
  \emph{Proc. IEEE}, vol.~98, no.~6, pp. 913--924, Jun. 2010.

\bibitem{Donoho-09PNAS}
{D. L. Donoho, A. Maleki, and A. Montanari}, ``{Message passing algorithms for
  compressed sensing},'' \emph{Proc. Nat. Acad. Sci.}, 2009.

\bibitem{Kabashima-09JSM}
{Y. Kabashima, T. Wadayama, and T. Tanaka}, ``{A typical reconstruction limit
  for compressed sensing based on lp-norm minimization},'' \emph{J. Stat.
  Mech.}, no.~9, p. L09003, 2009.

\bibitem{Ganguli-10PRL}
{S. Ganguli and H. Sompolinsky}, ``{Statistical mechanics of compressed
  sensing},'' \emph{Phys. Rev. Lett.}, vol. 104, no.~18, p. 188701, May 2010.

\bibitem{Rangan-12TIT}
{S. Rangan, A. K. Fletcher, and V. K. Goyal}, ``{Asymptotic analysis of MAP
  estimation via the replica method and applications to compressed sensing},''
  \emph{IEEE Trans. Inf. Theory}, vol.~58, no.~3, pp. 1902--1923, Mar. 2012.

\bibitem{Do-08ICASSP}
{T. T. Do, T. D. Tran, and L. Gan}, ``{Fast compressive sampling with
  structurally random matrices},'' in \emph{Proc. IEEE Int. Conf. Acoustics,
  Speech and Signal Processing (ICASSP)}, Las Vegas, NV, 2008, pp. 3369--3372.

\bibitem{Barbier-13ArXiv}
\BIBentryALTinterwordspacing
{J. Barbier, F. Krzakala, and C. Sch\"{u}lke}, ``{Compressed sensing and
  approximate message passing with spatially-coupled Fourier and Hadamard
  matrices},'' preprint, 2013. [Online]. Available:
  \url{http://arxiv.org/abs/1312.1740.}
\BIBentrySTDinterwordspacing

\bibitem{Javanmard-12ISIT}
{A. Javanmard and A. Montanari}, ``{Subsampling at information theoretically
  optimal rates},'' in \emph{Proc. IEEE Int. Symp. Information Theory (ISIT)},
  2012, pp. 2431--2435.

\bibitem{Donoho-09MPS}
{D. L. Donoho and J. Tanner}, ``{Observed universality of phase transitions in
  high-dimensional geometry, with implications for modern data an alysis and
  signal processing},'' \emph{Philos. Trans. Roy. Soc. London A, Math. Phys.
  Sci.}, vol. 367, no. 1906, pp. 4273--4293, Nov. 2011.

\bibitem{Bayati-12ArXiv}
\BIBentryALTinterwordspacing
{M. Bayati, M. Lelarge, and A. Montanari}, ``{Universality in polytope phase
  transitions and message passing algorithms},'' preprint, 2012. [Online].
  Available: \url{http://arxiv.org/abs/1207.7321.}
\BIBentrySTDinterwordspacing

\bibitem{Kabashima-12JSM}
{Y. Kabashima, M. M. Vehkaper\"{a}, and S. Chatterjee}, ``{Typical l1-recovery
  limit of sparse vectors represented by concatenations of random orthogonal
  matrices},'' \emph{J. Stat. Mech.}, vol. 2012, no.~12, p. P12003, 2012.

\bibitem{Tulino-13IT}
{A. M. Tulino, G. Caire, S. Verd\'{u}, and S. Shamai}, ``{Support recovery with
  sparsely sampled free random matrices},'' \emph{IEEE Trans. Inf. Theory},
  vol.~59, no.~7, pp. 4243--4271, Jul. 2013.

\bibitem{Vehkapera-arXiv13}
\BIBentryALTinterwordspacing
{M. Vehkaper\"{a}, Y. Kabashima, and S. Chatterjee}, ``{Analysis of regularized
  LS reconstruction and random matrix ensembles in compressed sensing},''
  preprint, 2013. [Online]. Available: \url{http://arxiv.org/abs/1312.0256.}
\BIBentrySTDinterwordspacing

\bibitem{Takeda-06EPL}
{K. Takeda, S. Uda, and Y. Kabashima}, ``{Analysis of CDMA systems that are
  characterized by eigenvalue spectrum},'' \emph{Europhysics Letters}, vol.~76,
  pp. 1193--1199, 2006.

\bibitem{Hatabu-09PRE}
{A. Hatabu, K. Takeda, and Y. Kabashima}, ``{Statistical mechanical analysis of
  the Kronecker channel model for multiple-input multipleoutput wireless
  communication},'' \emph{Phys. Rev. E}, vol.~80, pp. 061\,124(1--12), 2009.

\bibitem{Kitagawa-10CN}
{K. Kitagawa and T. Tanaka}, ``{Optimization of sequences in CDMA systems: A
  statistical-mechanics approach},'' \emph{Computer Networks}, vol.~54, no.~6,
  pp. 917--924, 2010.

\bibitem{Oymak-14ISIT}
{S. Oymak and B. Hassibi}, ``{A case for orthogonal measurements in linear
  inverse problems},'' in \emph{Proc. IEEE Int. Symp. Information Theory
  (ISIT)}, Honolulu, HI, 2014, pp. 3175--3179.

\bibitem{Bajwa-10Proc}
{W. U. Bajwa, J. Haupt, A. M. Sayeed, and R. Nowak}, ``{Compressed channel
  sensing: A new approach to estimating sparse multipath channels},''
  \emph{Proc. IEEE}, vol.~98, no.~6, pp. 1058--1076, Jun. 2010.

\bibitem{Mitra-DSP}
S.~K. Mitra, \emph{{Digital Signal Processing: A Computer-Based
  Approach}}.\hskip 1em plus 0.5em minus 0.4em\relax McGraw-Hill
  Science/Engineering/Math, 4 edition, 2010.

\bibitem{Fowler-12FtSP}
S.~M. J.~E.~Fowler and E.~W. Tramel, \emph{{Block-based compressed sensing of
  images and video}}.\hskip 1em plus 0.5em minus 0.4em\relax Foundations and
  Trends in Signal Processing, 2012.

\bibitem{Tanaka-10ISIT}
{T. Tanaka and J. Raymond}, ``{Optimal incorporation of sparsity information by
  weighted l1-optimization},'' in \emph{Proc. IEEE Int. Symp. Information
  Theory (ISIT)}, Jun. 2010, pp. 1598--1602.

\bibitem{Krzakala-12JSM}
{F. Krzakala, M. M\'{e}zard, F. Sausset, Y. Sun, and L. Zdeborov\'{a}},
  ``{Probabilistic reconstruction in compressed sensing: algorithms, phase
  diagrams, and threshold achieving matrices},'' \emph{J. Stat. Mech.}, vol.
  P08009, 2012.

\bibitem{Poor-94BOOK}
H.~V. Poor, \emph{{An Introduction to Signal Detection and Estimation}}.\hskip
  1em plus 0.5em minus 0.4em\relax New York: Springer-Verlag, 1994.

\bibitem{Edwards-75JPF}
{S. F. Edwards and P.W. Anderson}, ``{Theory of spin glasses},'' \emph{J.
  Physics F: Metal Physics}, vol.~5, pp. 965--974, 1975.

\bibitem{Nishimori-01BOOK}
H.~Nishimori, \emph{{Statistical Physics of Spin Glasses and Information
  Processing: An Introduction}}.\hskip 1em plus 0.5em minus 0.4em\relax ser.
  Number 111 in Int. Series on Monographs on Physics. Oxford U.K.: Oxford Univ.
  Press, 2001.

\bibitem{Tanaka-02IT}
{T. Tanaka}, ``{A statistical-mechanics approach to large-system analysis of
  CDMA multiuser detectors},'' \emph{IEEE Trans. Inf. Theory}, vol.~48, no.~11,
  pp. 2888--2910, Nov. 2002.

\bibitem{Moustakas-03TIT}
{A. L. Moustakas, S. H. Simon, and A. M. Sengupta}, ``{MIMO capacity through
  correlated channels in the presence of correlated interferers and noise: a
  (not so) large N analysis},'' \emph{IEEE Trans. Inf. Theory}, vol.~49,
  no.~10, pp. 2545--2561, Oct. 2003.

\bibitem{Guo-05IT}
{D. Guo and S. Verd\'{u} }, ``{Randomly spread CDMA: asymptotics via
  statistical physics},'' \emph{IEEE Trans. Inf. Theory}, vol.~51, no.~1, pp.
  1982--2010, Jun. 2005.

\bibitem{Muller-03TSP}
{R. R. M\"{u}ller}, ``{Channel capacity and minimum probability of error in
  large dual antenna array systems with binary modulation},'' \emph{IEEE Trans.
  Signal Process.}, vol.~51, no.~11, pp. 2821--2828, Nov. 2003.

\bibitem{Wen-07TCOM}
{C. K. Wen, K. K. Wong, and J. C. Chen}, ``{Spatially correlated MIMO
  multiple-access systems with macrodiversity: Asymptotic analysis via
  statistical physics},'' \emph{IEEE Trans. Commun.}, vol.~55, no.~3, pp.
  477--488, Mar. 2007.

\bibitem{Takeuchi-13TIT}
{K. Takeuchi, R. R. M\"{u}ller, M. Vehkaper\"{a}, and T. Tanaka}, ``{On an
  achievable rate of large Rayleigh block-fading MIMO channels with no CSI},''
  \emph{IEEE Trans. Inf. Theory}, vol.~59, no.~10, pp. 6517--6541, Oct. 2013.

\bibitem{Girnyk-14TWC}
{M. A. Girnyk, M. Vehkaper\"{a}, L. K. Rasmussen}, ``{Large-system analysis of
  correlated MIMO channels with arbitrary signaling in the presence of
  interference},'' \emph{IEEE Trans. Wireless Commun.}, vol.~13, no.~4, pp.
  1536--1276, Apr. 2014.

\bibitem{cvx}
M.~Grant and S.~Boyd, ``{CVX}: Matlab software for disciplined convex
  programming, version 2.1,'' \url{http://cvxr.com/cvx}, Mar. 2014.

\bibitem{Farrell-11JFAA}
{B. Farrell}, ``{Limiting empirical singular value distribution of restrictions
  of discrete fourier transform matrices},'' \emph{Journal of Fourier Analysis
  and Applications}, vol.~17, no.~4, pp. 733--753, 2011.

\bibitem{Bai-10}
Z.~Bai and J.~W. Silverstein, \emph{Spectral Analysis of Large Dimensional
  Random Matrices}.\hskip 1em plus 0.5em minus 0.4em\relax Springer Series in
  Statistics, 2010.

\bibitem{Wen-14ArXiv}
\BIBentryALTinterwordspacing
{C.-K. Wen and K.-K. Wong}, ``{Analysis of compressed sensing with
  spatially-coupled orthogonal matrices},'' preprint, 2014. [Online].
  Available: \url{http://arxiv.org/abs/1402.3215.}
\BIBentrySTDinterwordspacing

\bibitem{Rangan-10ArXiv}
{S. Rangan}, ``{Generalized approximate message passing for estimation with
  random linear mixing},'' in \emph{Proc. IEEE Int. Symp. Information Theory
  (ISIT)}, Saint Petersburg, Russia, Aug. 2011, pp. 2168--2172.

\bibitem{Cakmak-14ArXiv}
\BIBentryALTinterwordspacing
{B. \c{C}akmak, O. Winther, and B. H. Fleury}, ``{S-AMP: Approximate message
  passing for general matrix ensembles},'' preprint, 2014. [Online]. Available:
  \url{http://arxiv.org/abs/1405.2767.}
\BIBentrySTDinterwordspacing

\bibitem{Ma-15SPL}
{J. Ma, X. Yuan, and L. Ping}, ``{Turbo compressed sensing with partial DFT
  sensing matrix},'' \emph{IEEE Signal Process. Lett.}, vol.~22, no.~2, pp.
  158--161, Feb. 2015.

\bibitem{Kabashima-08JPCS}
{Y. Kabashima}, ``{Inference from correlated patterns: a unified theory for
  perceptron learning and linear vector channels},'' \emph{J. Phys. Conf.
  Ser.}, vol.~95, no.~1, 2008.

\end{thebibliography}
\end{document}